\documentclass[reqno]{amsart}
\usepackage{float}
\usepackage{hyperref}
\usepackage{mathtools,amsmath,amssymb,mathrsfs}
\usepackage{graphicx}
\usepackage{tikz}
\usetikzlibrary{calc}
\usepackage{xcolor}
\usetikzlibrary{snakes}
\usepackage{amsthm}
\usepackage{enumitem}
\usetikzlibrary{arrows.meta,snakes} 
\graphicspath{{.}{figures/}}
\newcommand{\R}{\mathbb{R}}
\newcommand{\ha}{\frac{1}{2}}
\newcommand{\e}{\varepsilon}

\newcommand{\BBB}{\color{blue}} 
\newcommand{\EEE}{\color{black}} 
\renewcommand{\BBB}{}
\renewcommand{\EEE}{}

\theoremstyle{plain}
\begingroup
\newtheorem{theorem}{Theorem}[section]
\newtheorem{lemma}[theorem]{Lemma}
\newtheorem{proposition}[theorem]{Proposition}

\endgroup

\theoremstyle{definition}
\begingroup
\newtheorem{definition}[theorem]{Definition}
\newtheorem{remark}[theorem]{Remark}
\newtheorem{example}[theorem]{Example}
\endgroup

\newenvironment{step}[1]{\textbf{Step #1.}}{}

\usepackage[english]{babel}
\begin{document}
 
\title[Crystallinity of the homogenized energy density]{Crystallinity of the homogenized energy density of periodic lattice systems}

\author{Antonin Chambolle}
\address[Antonin Chambolle]{CEREMADE, CNRS, Universit\'e Paris-Dauphine,
  PSL, France}
\email{antonin.chambolle@ceremade.dauphine.fr}

\author{Leonard Kreutz}
\address[Leonard Kreutz]{WWU M\"unster, Germany}
\email{lkreutz@uni-muenster.de}

\keywords{$\Gamma$-convergence, Ising system, Crystallinity, Wulff Shape.}


\begin{abstract}
  We study the homogenized energy densities of periodic ferromagnetic Ising systems. We prove that, for finite range interactions, the homogenized energy density, identifying the effective limit, is crystalline, i.e. its Wulff crystal is a polytope, for which we can (exponentially) bound the number of vertices. This is achieved by deriving a dual representation of the energy density through a finite cell formula. This formula also \BBB permits \EEE easy
  numerical computations: we show a few experiments where we compute periodic
  patterns which minimize the anisotropy of the surface tension.
\end{abstract}

\subjclass[2010]{35B27, 49J45, 82B20, 82D40.}  
\maketitle

\setcounter{tocdepth}{1}

\maketitle

\section{Introduction}

The study of discrete interfacial energies has attracted widespread attention in the mathematical community over last decades, with applications in various contexts such as computer vision \cite{BlaZis}, crystallization problems \cite{BlancLewin:15}, fracture mechanics \cite{BacBraCIc,BraLewOrt,FriedrichSchmidt}, or statistical physics \cite{Pre,Sta}. To give examples, in computer vision the understanding of these energies allows to investigate functional correctness of segmentation algorithms \cite{Cha99}. Whereas for crystallization problems it gives fluctuation estimates on the macroscopic shape of the crystal cluster of ground state configurations \cite{CicLeo,FriedrichKreutz:19,FriedrichKreutz:20,MaiSch}. 

In this work, we consider energies defined on discrete periodic sets  $\mathcal{L} \subset \mathbb{R}^d$ and corresponding Ising systems. We refer to \cite{AliBraCic, BraCic,  CafDLL01,CafDLL05,CozDipVal,Daneri,GiuLebLie,GiuLieSer-11,GiuSer} for an abundant literature on the derivation of continuum limits of such systems and their effective behavior. More precisely,  we consider  $\mathcal{L}$  satisfying the following two conditions (see Figure~\ref{fig:L})
\begin{itemize}
\item[(i)] (Discreteness) There exists $c>0$ such that $\mathrm{dist}(x,\mathcal{L}\setminus\{x\}) \geq c$ for all $x \in \mathcal{L}$;
\item[(ii)] (Periodicity) There exists $T \in \mathbb{N}$ such that for all $z\in \mathbb{Z}^d$, it holds that $\mathcal{L}+Tz =\mathcal{L}$;
\end{itemize}
\begin{figure}[htb]
\begin{tikzpicture}[scale=.4]
\draw[<->](11,0)--++(0,4);
\draw(11,2) node[anchor=west]{$T$};

\draw[<->,white](-7,0)--++(0,4);
\draw[white](-7,2) node[anchor=east]{$T$};

\clip (-6,-6) rectangle (10,6);

\foreach \j in {-2,...,3}{

\foreach \k in {-2,...,3}{

\draw (-4+4*\j,-4+4*\k) rectangle (4+4*\j,4+4*\k);

\draw[fill=black] (-1+4*\j,-1+4*\k) circle(.05);
\draw[fill=black] (-2+4*\j,-1.2+4*\k) circle(.05);
\draw[fill=black] (-3+4*\j,-2+4*\k) circle(.05);
\draw[fill=black] (2+4*\j,1+4*\k) circle(.05);
\draw[fill=black] (3+4*\j,1.5+4*\k) circle(.05);
\draw[fill=black] (2.8+4*\j,1.9+4*\k) circle(.05);
\draw[fill=black] (2.2+4*\j,1.3+4*\k) circle(.05);
\draw[fill=black] (-0.2+4*\j,-0.2+4*\k) circle(.05);
\draw[fill=black] (1.8+4*\j,1.8+4*\k) circle(.05);
\draw[fill=black] (-3.8+4*\j,-3.6+4*\k) circle(.05);
\draw[fill=black] (-2.6+4*\j,-2.5+4*\k) circle(.05);
\draw[fill=black] (.8+4*\j,-1.2+4*\k) circle(.05);
\draw[fill=black] (.7+4*\j,1.2+4*\k) circle(.05);
\draw[fill=black] (-1.8+4*\j,-2.2+4*\k) circle(.05);
}
}
\end{tikzpicture}
\caption{An example of the set $\mathcal{L}$}
\label{fig:L}
\end{figure}
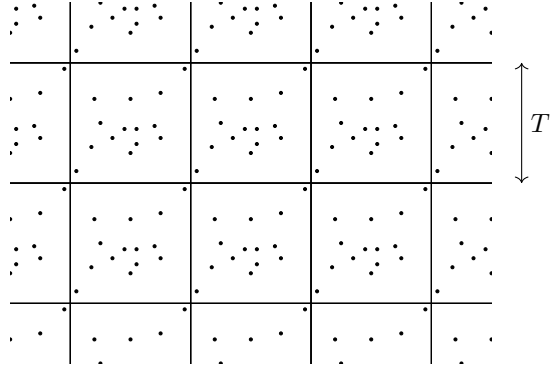 

To each function $u : \mathcal{L}\to \{0,1\}$ and each $A \subset \mathbb{R}^d$ we associate an energy
\begin{align}\label{intro:def energy}
E(u,A) = \sum_{i \in \mathcal{L}\cap A}\sum_{j \in \mathcal{L}} c_{i,j}(u(i)-u(j))^+\,,
\end{align}
where $(z)^+$ denotes the positive part of $z\in \mathbb{R}$, $c_{i,j} \colon \mathcal{L}\times \mathcal{L} \to [0,+\infty)$ are $T$-periodic, that is $c_{i+Tz,j+Tz} = c_{i,j}$ for all $i,j \in \mathcal{L}$ and $z \in \mathbb{Z}^d$ and satisfy the following decay assumption
\begin{itemize}
\item[(iii)] (Decay of interactions) For all $i \in \mathcal{L}$ there holds
\begin{align*}
\sum_{j \in \mathcal{L}} c_{i,j}|i-j| < +\infty\,.
\end{align*}
\end{itemize}
Assuming conditions (i)-(iii) (and some additional coercivity assumption) ensures that the asymptotic behavior of \eqref{intro:def energy} is well described (in a variational sense) by a continuum perimeter energy. More precisely, let us introduce a scaling parameter $\varepsilon >0$. We consider the scaled energies 
\begin{align*}
E_\varepsilon(u) = \sum_{i,j \in \varepsilon\mathcal{L}} \varepsilon^{d-1} c_{i,j}^\varepsilon(u(i)-u(j))^+\,,
\end{align*}
where  $c_{i,j}^\varepsilon = c_{i/\varepsilon,j/\varepsilon}$ and $u \colon\varepsilon \mathcal{L}\to \{0,1\}$. By identifying $u$ with its piecewise constant interpolation taking the value $u(i)$ on the Voronoi cell centered at $i \in \varepsilon\mathcal{L}$ we may regard the energies as defined on $L^1_{\mathrm{loc}}(\mathbb{R}^d, \{0,1\})$. Integral representation results \cite{AliCicRuf,AG,BraPiat} then guarantee that the energies $E_\varepsilon$ $\Gamma$-converge (see \cite{Braides:02, DalMaso:93} for an introduction to that subject) with respect to the $L^1_{\mathrm{loc}}(\mathbb{R}^d)$-topology to a continuum energy of the form
\begin{align*}
E_0(u) = \int_{\partial^*\{u=1\}} \varphi(\nu_u(x)) \,\mathrm{d}\mathcal{H}^{d-1}\,\quad u \in BV_{\mathrm{loc}}(\mathbb{R}^d;\{0,1\})\,.
\end{align*}
Here, $BV_{\mathrm{loc}}(\mathbb{R}^d;\{0,1\})$ denotes the space of functions with (locally) bounded variation and values in $\{0,1\}$, $\partial^*\{u=1\}$ denotes the reduced boundary of the level set $\{u=1\}$, $\nu_u(x)$ its measure theoretic normal at the point $x \in \partial^*\{u=1\}$, and $\mathcal{H}^{d-1}$ denotes the $(d-1)$-dimensional Hausdorff measure, see \cite{AFP} for the precise definitions of these notions. The energy density $\varphi \colon \mathbb{R}^d \to [0,+\infty)$ can be  recovered via the asymptotic cell formula
\begin{align}\label{intro:def varphi}
\varphi(\nu) := \lim_{\delta \to 0}\lim_{S\to +\infty} \frac{1}{S^{d-1}} \inf\left\{ E(u,Q^\nu_S)\colon u \colon \mathcal{L}\to\{0,1\}, u(i) = u_\nu(i) \text{ on } \mathcal{L} \setminus Q^\nu_{(1-\delta)
S}\right\}\,,
\end{align}
where 
\begin{align*}
u_\nu(x) = \begin{cases} 1&\text{if } \langle x,\nu \rangle\geq 0\,,\\
0&\text{otherwise.}
\end{cases}
\end{align*}
Here, \BBB $Q^\nu_S$ is a suitable rotation of the coordinate cube with side-length $S$  such that two faces are parallel to $\{\nu=0\}$\EEE.  In the case $\mathcal{L}=\mathbb{Z}^2$, $c_{i,j}=1$ if $|i-j|=1$ and $c_{i,j}=0$ otherwise, we have that $\varphi(\nu) =2\| \nu\|_1$, see Figure~\ref{Fig.L1}.
\begin{figure}[htb]
\begin{tikzpicture}[scale=.5]
\draw(3,0) node[anchor=north]{$\nu_1$};
\draw(0,3) node[anchor=east]{$\nu_2$};
\draw[->](0,-3)--++(0,6);
\draw[->](-3,0)--++(6,0);
\draw(2,0)--(0,2)--(-2,0)--(0,-2)--(2,0);
\draw(2.75,1.5) node{$\{\|\nu\|_1\leq 1\}$};
\end{tikzpicture}
\caption{The energy density in the case of nearest neighbor interactions on $\mathbb{Z}^2$}
\label{Fig.L1}
\end{figure}
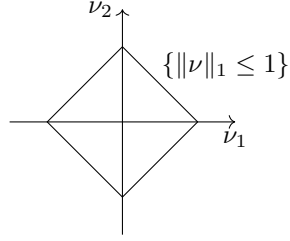

The goal of this article is to investigate the energy density $\varphi$. In particular we show, that for finite interaction range $c_{i,j}$, that is there exists $R>0$ such that $c_{i,j} = 0$ if $|i-j|>R$, then $\varphi$ is \BBB \emph{crystalline}, see \cite[Definition 3.2]{Fonseca}. \EEE This means that the solution to 
 \begin{align*}
 \min\left\{\int_{\partial^*A} \varphi(\nu_A(x))\,\mathrm{d}\mathcal{H}^{d-1} : |A|=1\right\}
\end{align*}  
is a \BBB convex polytope\EEE.
 The finite range of interaction is crucial. Indeed, example \ref{example:without(H3)} shows that for infinite range interactions this is in general not true. In \cite{BraKre,BraKre2} it is shown that, as the periodicity $T$ of the interactions tends to $+\infty$, it is possible to approximate any norm as surface energy density satisfying suitable growth conditions. We refer to \cite{AliCicRuf} for a random setting where it is shown that an isotropic energy density (and thus non-crystalline) can be obtained in the limit.

 The proof of the crystallinity in the case of finite range interactions relies on the following alternative representation \BBB formula of the density\EEE, proven in Proposition \ref{prop:main}. Namely, we prove that
 \begin{align}\label{formula alt:varphi}
\varphi(\nu) = \frac{1}{T^d}\inf\left\{E(u,Q_T)\colon u\colon \mathcal{L} \to \mathbb{R},  u(\cdot)-\langle\nu,\cdot\rangle \text{ is } T\text{-periodic}\right\}\,.
\end{align} 
This representation formula is reminiscent of the representation formula of the energy density of integral functionals obtained via homogenization of $T$-periodic integral functionals in $W^{1,p}$ \cite{ADF}. To motivate this, consider the positively $1$-homogeneous extensions of $E_\varepsilon$ defined by
\begin{align*}
F_\varepsilon(u) = \sum_{i,j \in \varepsilon\mathcal{L}} \varepsilon^{d-1} c_{i,j}^\varepsilon(u(i)-u(j))^+\,,
\end{align*}
for $u : \varepsilon\mathcal{L} \to \mathbb{R}$. The $\Gamma$-limit $F_0$ of the above sequence is clearly positively $1$-homogeneous and convex as the sequence of functionals satisfies these properties. Thus, $F_0$ admits an integral representation of the form
\begin{align*}
F_0(u) = \int f_0(\nabla u) \,\mathrm{d}x + \int f_0\left(\frac{dD_su}{d|D_su|}\right)\,\mathrm{d}|D_su|\,,
\end{align*}
where $f_0 \colon \mathbb{R}^d\to \mathbb{R}^d$ is convex and positively $1$-homogeneous, see \cite{BraidesChiado}. (We like to stress however, that this integral representation for the spin energies considered above is not proven in the literature.) Here, the important point is that the density of the singular part and the density of absolutely continuous parts agree. In the continuous setting, in \cite{ChaGiaLus,ChaTho} it has been shown that for continuous and convex densities, that satisfy a coarea formula, the $\Gamma$-convergence of sets of finite perimeter or in the space of $BV$-functions is equivalent. Thus also in their setting, the densities agree.  The density of the absolutely continuous part can be calculated via \eqref{formula alt:varphi}.  This property eventually allows us to express $\varphi$ via \eqref{formula alt:varphi} since the density of the absolutely continuous part can be calculated via \eqref{formula alt:varphi} and the density of the
singular part agrees with the energy density in \eqref{intro:def varphi}, see Proposition \ref{prop:main}. Using convex duality (see \cite{Rockafellar}) and using \eqref{formula alt:varphi} we show in Theorem \ref{theorem:main} that $\varphi$ is crystalline, and estimate an upper bound on the number of
extreme points of the corresponding Wulff shape. We would like to stress that \eqref{formula alt:varphi} is not only a useful tool in our proof but it can be used also for computational purposes as it is a finite and not an asymptotic cell formula.

The paper is organized as follows. In Section \ref{sec:Setting} we describe the mathematical setting and state the main theorems of our paper. In Section \ref{sec:proofofprop} we prove Proposition \ref{prop:main}, the alternate representation formula for $\varphi$. In Section \ref{sec:Crystallinity} we show that, in the case of finite range interactions, the density $\varphi$ is always crystalline. \BBB In Section \ref{sec:dif} we discuss some differentiability properties of $\varphi$. \EEE  We present some numerical simulations of our findings in the last chapter.

\section{Setting of the problem and statement of the main result}\label{sec:Setting}

\subsection{Notation}\label{subsec:notation}
We denote by $\mathcal{B}(\mathbb{R}^d)$ the collection of all Borel-Sets in $\mathbb{R}^d$.
For every $A \subset \mathbb{R}^d$ we denote by $|A|$ its $d$-dimensional Lebesgue measure. Given $r>0$, we denote by $(A)_r := \{x\in \mathbb{R}^d : \mathrm{dist}(x,A) <r\}$ the $r$-neighbourhood of $A$. Given $\tau \in \mathbb{R}^d$, we set $A+\tau := \{x+\tau : x\in A\}$. The set $\mathbb{S}^{d-1}:=\{\nu \in \mathbb{R}^d: |\nu|=1\}$ is the set of all $d$-dimensional unit vectors. For $v,w \in \mathbb{R}^d$ we denote by $\langle v,w\rangle$ their scalar product. We denote by $\{e_1,\ldots,e_d\}\subset \mathbb{R}^d$ the standard orthonormal basis of $\mathbb{R}^d$. Given $C \subset \mathbb{R}^d$ convex, we denote by $\mathrm{extreme}(C)$ its extreme points. Given $\rho>0$, we denote by $Q_\rho :=[-\rho/2,\rho/2)^d$ the half open cube centred in $0$ with side-length $\rho$. For $\nu \in \mathbb{S}^{d-1}$, we set $Q^\nu_\rho := R^\nu Q_\rho$, where $R^\nu$ is a rotation such that $R^\nu e_d =\nu$. Furthermore, given $x \in \mathbb{R}^d$ we set $Q^\nu_\rho(x) := x+ Q^\nu_\rho$ (resp.~$Q_\rho(x) = x+ Q_\rho$). Given $x\in \mathbb{R}^d$ and $r>0$ we denote by $B_r(x)$ the open ball with radius $r>0$  and center $x$. For $A \subset \mathbb{R}^d$ we denote by $\chi \colon \mathbb{R}^d\to \{0,1\}$ the characteristic function of the set $A$ given by
\begin{align}\label{def:characteristic}
\chi_A(x) := \begin{cases} 1 &\text{if } x\in A\,,\\
0&\text{otherwise.}
\end{cases}
\end{align}
 We denote by $\omega_d$ the volume of the unit ball in $\mathbb{R}^d$.
 Given $\nu \in \mathbb{S}^{d-1}$ we define
\begin{align}\label{def:unu}
u_\nu(x) := \begin{cases} 1 &\text{if } \langle\nu,x\rangle \geq 0\,,\\
0&\text{otherwise.}
\end{cases}
\end{align}
For $z\in \mathbb{R}$ we denote by $(z)^+:= \max\{z,0\}$ the positive part of $z$.

\subsection{Discrete energies and homogenized surface energy density}\label{subsec:discrete energies} In this paragraph we define the discrete energies we want to consider and the homogenized surface energy density.

Let $\mathcal{L} \subset \mathbb{R}^d$ satisfy the following two conditions:

\BBB {\rm (L1)} (Discreteness) \EEE There exists  $c>0$ such that for all $x\in \mathcal{L}$ there holds
\begin{align*}
\mathrm{dist}(x,\mathcal{L}\setminus\{x\}) \geq c\,.
\end{align*}

\smallskip

\BBB {\rm (L2)} (Periodicity) \EEE There exists $T \in \mathbb{N}$ such that for all $z\in \mathbb{Z}^d$ there holds
\begin{align*}
\mathcal{L} +Tz = \mathcal{L}\,.
\end{align*}

Note that the two assumptions \BBB {\rm (L1)} \EEE and \BBB {\rm (L2)} \EEE include multi-lattices, such as the hexagonal closed packing lattice in three dimensions, and \BBB Bravais \EEE lattices, such as $\mathbb{Z}^d$, or the \BBB face-centered \EEE cubic lattice in three dimensions.

We consider interaction coefficients $c_{i,j}: \mathcal{L}\times\mathcal{L} \to [0,+\infty)$ and the corresponding (localized) ferromagnetic spin energies of the form
\begin{align}\label{def:E}
E(u,A):= \sum_{i \in \mathcal{L}\cap A} \sum_{j \in \mathcal{L}} c_{i,j} (u(i)-u(j))^+\,,
\end{align}
where $u \colon \mathcal{L}\to \mathbb{R}$ and $A \in \mathcal{B}(\mathbb{R}^d)$. If $A=\mathbb{R}^d$ we omit the dependence on the set and write $E(u):= E(u,\mathbb{R}^d)$. We want to remark that we are considering interactions on the directed graph instead of the undirected graph. 

We introduce the following three \BBB hypotheses \EEE on the interaction coefficients $c_{i,j} \colon \mathcal{L}\times \mathcal{L} \to [0,+\infty)$:

\smallskip

{\rm (H1)} (Periodicity) There holds
\begin{align*}
c_{i+Tz,j+Tz} = c_{i,j} 
\end{align*}
for all $i,j \in \mathcal{L}$, $z \in \mathbb{Z}^d$.

\smallskip 

{\rm (H2)} (Decay of Interactions) For all $i \in \mathcal{L}$ there holds
\begin{align*}
\sum_{j \in \mathcal{L}} c_{i,j}|i-j| < +\infty\,.
\end{align*}

\smallskip

{\rm (H3)} (Finite Range Interactions) There exists $R >1$ such that 
\begin{align*}
c_{i,j} = 0
\end{align*}
for all $i,j\in \mathcal{L}$ such that $|i-j| \geq R$.

It is obvious, that hypothesis {\rm (H3)} implies hypothesis {\rm (H2)}. Note that, if {\rm (H1)} and {\rm (H2)} are satisfied then
\begin{align*}
\max_{i \in \mathcal{L}} \sum_{j \in \mathcal{L}} c_{i,j}|i-j| = \max_{i \in \mathcal{L}\cap Q_T} \sum_{j \in \mathcal{L}} c_{i,j}|i-j| <+\infty
\end{align*} 
and for all $R>0$, there exists $C_R>0$ such that $C_R\to 0$ as $R\to +\infty$ and
\begin{align*}
\max_{i \in \mathcal{L}} \underset{|i-j|\geq R}{\sum_{j \in \mathcal{L}}} c_{i,j}|i-j| \leq C_R\,.
\end{align*}

\begin{definition}\label{def:homogenized energy density} Let $c_{i,j}$ satisfy {\rm (H1)} and {\rm (H2)}. We then define the \emph{homogenized surface energy density} $\varphi : \mathbb{R}^d \to [0,+\infty)$ as the convex positively homogeneous function of degree one such that for all $\nu \in \mathbb{S}^{d-1}$ we have
\begin{align}\label{def:varphi}
\varphi(\nu) := \lim_{\delta \to 0}\lim_{S\to +\infty} \frac{1}{S^{d-1}} \inf\left\{ E(u,Q^\nu_S) \colon u \colon \mathcal{L}\to\{0,1\}, u(i) = u_\nu(i) \text{ on } \mathcal{L} \setminus \BBB Q^\nu_{(1-\delta)
S}\right\}\,,\EEE
\end{align}
with $u_\nu$ defined in \eqref{def:unu}.
\end{definition}

\begin{remark}\label{rem:discretetocontinuum} The definition above can be interpreted as a passage from discrete to continuum description as follows. Given $\varepsilon>0$, we consider the scaled energies
\begin{align*}
E_\varepsilon(u) := \sum_{i \in \mathcal{L}} \sum_{j \in \mathcal{L}} \varepsilon^{d-1} c_{i,j} (u(\varepsilon i)-u(\varepsilon j))^+\,,
\end{align*}
where $u : \varepsilon\mathcal{L}\to \{0,1\}$. Upon identifying $u$ with its piecewise-constant interpolation, we can regard these energies to be defined on $L^1_{\rm loc}(\mathbb{R}^d)$. We know that their $\Gamma$-limit is infinite outside the space \BBB $BV_{\rm loc}(\mathbb{R}^d;\{0,1\})$\EEE, where it has the form
\begin{align*}
E_0(u) := \int_{\partial^* \{u=1\}} \varphi(\nu)\,\mathrm{d}\mathcal{H}^{d-1}
\end{align*}
with $\varphi$ given by \eqref{def:varphi}, see for example \cite{AliCicRuf}.\footnote{Actually, the integral representation for the $\Gamma$-limit has only been shown for undirected graphs. However, a slight modification of the proof shows that it is still true for directed graphs.} Here, $\partial^*\{u=1\}$ denotes the reduced boundary of the set $\{u=1\}$  and $\mathcal{H}^{d-1}$ denotes the $(d-1)$-dimensional Hausdorff measure in $\mathbb{R}^d$ (cf.~\cite{AFP}, Chapters 2.8 and 3.5).
\end{remark}

\begin{remark}\label{rem:continuity} Testing with $u_\nu$ in \eqref{def:varphi}, using \BBB {\rm(L1)} \EEE and {\rm (H2)}, it is easy to see that $\varphi(\nu) \leq C$ for all $\nu \in \mathbb{S}^{d-1}$. Therefore, due to the convexity and the fact that it is a positively one homogeneous function of degree one, $\varphi$ is Lipschitz continuous. 
\end{remark}

\subsection{Statement of the main result}\label{subsec:main result} In this section we state the main result.

\begin{definition} \label{def:Wulff} Given $\varphi \colon \mathbb{R}^d \to [0,+\infty)$ convex, positively homogeneous of degree one, we define the Wulff set of $\varphi$ by
\begin{align*}
W_\varphi := \{ \zeta \in \mathbb{R}^d \colon \langle \zeta,\nu \rangle \leq \varphi(\nu) \text{ for all } \nu \in \mathbb{S}^{d-1}\}\,.
\end{align*}
We say that $\varphi$ is crystalline, if $W_\varphi$ is a polytope.
\end{definition}
\begin{remark}\label{rem:duality} From the definition of the Wulff set, it is clear that
\begin{align*}
\varphi(\nu) = \sup_{\zeta \in W_\varphi} \langle\nu,\zeta\rangle\,.
\end{align*}
Furthermore, one can check, that if $\varphi$ is crystalline, then the set $\{\varphi \leq 1\}$ is a polytope.
\end{remark}

The next proposition shows that, we obtain a finite cell formula in order to calculate $\varphi$ instead of the asymptotic one, given in \eqref{def:varphi}. We think that this result in itself is interesting, since it allows for calculations on finite size systems in order to compute $\varphi$ for general Ising systems. This result is in spirit very close to \cite{BouDal,BraidesChiado}, where convex and positively $1$-homogenous continuum energies are considered. In this case, the surface energy density and the energy with respect to the absolutely continuous part coincide.
\BBB
For $k \in \mathbb{N}$ let $\Lambda:= Tk$ and denote by
\begin{align}\label{def:T-periodic functions}
\mathcal{A}_\mathrm{per}(Q_\Lambda;\mathbb{R}):=\{u\colon \mathcal{L}\to \mathbb{R} \colon u(x+\Lambda z)=u(x) \text{ for all } z \in \mathbb{Z}^d\}
\end{align}
be the space of $\Lambda$-periodic functions.
\EEE
\begin{proposition}\label{prop:main} Let $c_{i,j} \colon \mathcal{L} \times \mathcal{L} \to [0,+\infty)$ be interaction coefficients such that {\rm (H1)} and {\rm (H2)} hold true. Then
\begin{align}\label{eq:propvarphi}
\varphi(\nu) = \frac{1}{T^d}\inf\left\{E(u,Q_T)\colon u\colon \mathcal{L} \to \mathbb{R},  u(\cdot)-\langle\nu,\cdot\rangle \BBB \in \mathcal{A}_\mathrm{per}(Q_T;\mathbb{R})\EEE\right\}\,.
\end{align} 
\end{proposition}

\begin{theorem}\label{theorem:main} Let $c_{i,j} \colon \mathcal{L}\times \mathcal{L} \to [0,+\infty)$ be interaction coefficients such that {\rm (H1)} and {\rm (H3)} hold true. Then, the homogenized surface energy density $\varphi$ is crystalline. Denote by
\begin{align*}
N:=\#\{(i,j) \in \mathcal{L} \cap Q_T \times \mathcal{L}  \colon c_{i,j} \neq 0 \}\,.
\end{align*}
 Then,
\begin{align*}
\#\mathrm{extreme}(W_{\varphi}) \leq 3^{N}\,.
\end{align*}
\end{theorem}
The next example shows that without assumption \BBB{\rm (H3)} \EEE Theorem \ref{theorem:main} fails to hold true.
\begin{example}\label{example:without(H3)}  To construct the example we first observe that if $f \colon\mathbb{R}^d \to [0,+\infty)$ is crystalline, then $D^2f$ is a Radon-measure with support contained in finitely many hyper-planes. To see this, note that if $f \colon\mathbb{R}^d \to [0,+\infty)$ is crystalline, then  there exist $ \{\xi_k\}_{k=1}^N \subset \mathbb{R}^d$ such that
\begin{align*}
f(\nu) = \max_{1\leq k \leq N} \langle\xi_k,\nu\rangle\,.
\end{align*}
Here, we assume that $ \{\xi_k\}_{k=1}^N$ is chosen minimal, \BBB i.e.~if we set $V_k =\left\{ \nu \in \mathbb{R}^d \colon f(\nu) = \langle\xi_k,\nu\rangle\right\}$, then $ |V_k| >0$ for all $k=1,\ldots,N$. \EEE This assumption ensures that all the vectors $\xi_k$ play an active role in the definition of $f$. Now, $Df \in BV_{\mathrm{loc}}(\mathbb{R}^d;\mathbb{R}^d)$ is given by
\begin{align*}
Df(\nu) =\sum_{k=1}^N \chi_{V_k}(\nu)\xi_k\,,
\end{align*}
\BBB with $V_k$ defined above. \EEE Then
\begin{align*}
D^2f(\nu) = \sum_{1\leq k < j \leq N} (\xi_k-\xi_j)\otimes \nu_{kj} \mathcal{H}^{d-1}\lfloor_{\partial V_k \cap \partial V_j}\,,
\end{align*}
where 
$ \partial V_k \cap \partial V_j = \left\{ \nu \in \mathbb{R}^d \colon f(\nu) = \langle\xi_k,\nu\rangle=\langle\xi_j,\nu\rangle\right\}$ and $\nu_{kj} \in \mathbb{S}^{d-1}$ denotes the normal pointing towards the set $V_k$.

Let now $\mathcal{L} = \mathbb{Z}^d$ and $c_{i,j}=c_{j-i}=c_{i-j}$ (in the following denoted by $\{c_\xi\}_{\xi \in \mathbb{Z}^d}$) be such that $c_\xi >0$ for all $\xi \in \mathbb{Z}^d$ and 
\begin{align*}
\sum_{\xi \in \mathbb{Z}^d} c_\xi |\xi| <+\infty\,.
\end{align*}
It is then obvious that $c_{i,j}$ is $1$-periodic, (H1) and (H2) hold true, but (H3) is violated. Therefore, due to Proposition \ref{prop:main}, we have
\begin{align*}
\varphi(\nu) = \sum_{\xi \in \mathbb{Z}^d} c_\xi |\langle\xi,\nu\rangle|\,.
\end{align*}
This is true, since the only admissible functions in the minimum problem given by Proposition \ref{prop:main}  are $u_\nu(i) = \langle \nu,i\rangle+c$ for some $c \in \mathbb{R}$. We claim that
\begin{align*}
D\varphi(\nu) = \sum_{\xi \in \mathbb{Z}^d} \mathrm{sign}(\langle\xi, \nu\rangle) \, c_\xi \,\xi\,, 
\end{align*}
where $\mathrm{sign} \colon \mathbb{R} \to \mathbb{R}$ is defined by 
\begin{align*}
\mathrm{sign}(t) = \begin{cases} 1 & t \geq 0;\\
-1&t<0.
\end{cases}
\end{align*}
 Therefore 
\begin{align}\label{eq:D2varphi}
D^2\varphi = 2\sum_{\xi \in \mathbb{Z}^d} c_\xi \xi \otimes \frac{\xi}{|\xi|} \mathcal{H}^{d-1}\lfloor_{\{\nu\colon\langle\xi,\nu\rangle=0\}}\,.
\end{align}
This can be seen by approximation. Consider $\varphi_R :\mathbb{R}^d\to \mathbb{R}$ defined by
\begin{align*}
\varphi_R(\nu) = \underset{|\xi|\leq R}{\sum_{\xi \in \mathbb{Z}^d}} c_\xi |\langle\xi,\nu\rangle|\,,\quad D\varphi_R(\nu) = \underset{|\xi|\leq R}{\sum_{\xi \in \mathbb{Z}^d}} \mathrm{sign}(\langle\xi, \nu\rangle) \, c_\xi \,\xi\,, 
\end{align*}
Then
\begin{align*}
D^2\varphi_R =2\underset{|\xi| \leq R}{\sum_{\xi \in \mathbb{Z}^d}} c_\xi \xi \otimes \frac{\xi}{|\xi|} \mathcal{H}^{d-1}\lfloor_{\{\nu\colon \langle\xi,\nu\rangle=0\}}\,.
\end{align*}
Now
\begin{align*}
|D^2\varphi_R|(B_r) \leq Cr^{d-1} \sum_{\xi \in \mathbb{Z}^d} c_\xi |\xi|\,,
\end{align*}
so the total variation of $D^2\varphi_R$ is (locally) uniformly bounded with limiting measure $D^2\varphi$ and $D\varphi_R \to D\varphi$ in $L^1_{\mathrm{loc}}(\mathbb{R}^d;\mathbb{R}^d)$, actually weakly in $BV$. Hence, \eqref{eq:D2varphi} is shown. Now, since $c_\xi >0$ for all $\xi \in \mathbb{Z}^d$ it is obvious that $D^2\varphi$ is not supported on finitely many hyper-planes. Thus $\varphi$ cannot be crystalline. 

\BBB Note that $\varphi$ is differentiable in totally irrational directions.\footnote{$p$ is totally irrational if there is no $q\in \mathbb{Z}^d\setminus \{0\}$ such that $\langle q,p\rangle=0$.}
A similar property is known to hold, in the continuous setting~\cite{AuBa06,CGN14}, for homgenized surface tensions.
We can state a result showing that this is still the case in the discrete
setting, under assumptions (H1) and (H2).
\end{example}
\BBB
\begin{proposition}\label{prop:dif}
  Under the assumptions of Proposition~\ref{prop:main}, $\varphi$ is differentiable in any
  totally irrational direction.
\end{proposition}
It is expected that it should be, ``in general'', not differentiable in
the other directions, at least whenever the minimizers $u$ in~\eqref{eq:propvarphi} are constant
on an infinite set, however the proofs in~\cite{AuBa06,CGN14} rely on ellipticity
properties of the problem and are less easy to transfer to the discrete case.
The proof of Proposition~\ref{prop:dif}, which mimicks the proof in~\cite{CGN14},
is postponed to Section~\ref{sec:dif}, and
relies on the dual representation~\eqref{eq:dual representation} introduced later on. \EEE

\section{Proof of Proposition \ref{prop:main}}\label{sec:proofofprop}

This section is devoted to the proof of Proposition \ref{prop:main}.  We assume throughout this section that assumptions \BBB {\rm (L1)}, {\rm (L2)} \EEE and {\rm (H1)}, {\rm (H2)} are satisfied.  The proof consists in showing that $\varphi$ can be characterized by several (equivalent) cell-formulas and therefore passing from \eqref{def:varphi} to \eqref{eq:propvarphi}.

First, we will state and prove some elementary properties of $E$ that will be used throughout this section.

\begin{lemma}\label{lemma:elementaryproperties} \BBB Let $A \in \mathcal{B}(\mathbb{R}^d)$ and let $c>0$ be as in {\rm (L1)}.\EEE
\begin{enumerate}[label={\rm(\roman*)}]
\item \BBB There exists a universal constant $C>0$ (depending only on $c_{i,j}$ and $c$ in {\rm (L1)}) such that for all $\nu \in \mathbb{R}^d$ we have
\begin{align*}
E(\langle\nu,\cdot\rangle,A) \leq C|\nu| |(A)_c|\,.
\end{align*} 
\EEE
\item Let $u \colon \mathcal{L} \to \mathbb{R}$.  For all $ t\in \mathbb{R},\lambda >0$ there holds 
\begin{align*}
E(\lambda u+ t,A) = \lambda E(u, A)
\end{align*}
and $u \mapsto E(u,A)$ is convex. In particular, 
\begin{align*}
E(u+v,A) \leq E(u,A) +E(v,A)
\end{align*}
for all $u,v \colon \mathcal{L}\to \mathbb{R}$.
\item Let $u \colon \mathcal{L} \to \mathbb{R}$ and $B \in \mathcal{B}(\mathbb{R}^d)$ be such that $A\subset B$. Then
\begin{align*}
E(u,A) \leq E(u,B)\,.
\end{align*}
\item  Let $u \colon \mathcal{L}\to \mathbb{R}$ and $B \in \mathcal{B}(\mathbb{R}^d)$ be such that $A\cap  B = \emptyset$ . Then
\begin{align*}
E(u,A \cup B) = E(u,A) + E(u,B)\,.
\end{align*}
\item \BBB We have \EEE
\begin{align*}
\#\{i \in \mathcal{L} \cap A\} \leq \frac{1}{c^d \omega_d} |(A)_{c}|\,.
\end{align*}
\item Let $u \colon \mathcal{L}\to \mathbb{R}$. Then, the function $\tau \mapsto E(u(\cdot-\tau),A+\tau)$ is $T$-periodic.
\end{enumerate}
\end{lemma}
\begin{proof}  We start by proving {\rm (ii)}-{\rm (iv)} in Step 1, then {\rm (v)} and {\rm (vi)} in Step 2 and Step 3 respectively, and finally {\rm (i)} in Step 4.

\noindent
\begin{step}{1}(Proof of {\rm (ii)} -  {\rm (iv)}) All the three statements are a direct consequence of \eqref{def:E} and the fact that $c_{i,j} \geq 0$.
\end{step}

\noindent
\begin{step}{2} (Proof of {\rm (v)}) Note that
\begin{align*}
\bigcup_{i \in \mathcal{L} \cap A} B_c(i) \subset (A)_c
\end{align*}
 and therefore, \BBB due to {\rm (L1)}, \EEE
 \begin{align*}
c^d \omega_d\#\{i \in \mathcal{L} \cap A\}  =\left|\bigcup_{i \in \mathcal{L} \cap A} B_c(i)\right| \leq  |(A)_c|\,.
 \end{align*}
 This is the claim.
\end{step}

\noindent
\begin{step}{3}(Proof of {\rm (vi)}) Let $u \colon \mathcal{L}\to \mathbb{R}$ and $ z \in \mathbb{Z}^d$. Then, using {\rm (H1)} and \BBB {\rm (L2)}, \EEE
\begin{align*}
E(u(\cdot-Tz),A+Tz) &=\sum_{i \in \mathcal{L} \cap (A+Tz)}\sum_{j \in \mathcal{L}} c_{i,j}(u(i-Tz)-u(j-Tz))^+\\&= \sum_{i \in \mathcal{L}\cap A}\sum_{j \in (\mathcal{L}+Tz)} c_{i+Tz,j+Tz}(u(i)-u(j))^+ \\&= \sum_{i \in \mathcal{L}\cap A}\sum_{j \in \mathcal{L}} c_{i,j}(u(i)-u(j))^+ = E(u,A)\,.
\end{align*}
\end{step}
\noindent
\begin{step}{4}(Proof of {\rm (i)}) \BBB Let $\nu \in \mathbb{R}^d$, then, due to {\rm (v)}, \BBB {\rm (L1)}, {\rm (L2)},{\rm (H1)}, \EEE and {\rm (H2)}, we have \EEE
\begin{align*}
E(\langle\nu,\cdot\rangle,A) = \sum_{i \in \mathcal{L}\cap A} \sum_{j \in \mathcal{L}} c_{i,j} |\langle \nu, i-j\rangle| \leq |\nu| \#\{i \in \mathcal{L} \cap A\} \max_{i \in \mathcal{L} } \sum_{j \in \mathcal{L}} c_{i,j} |i-j| \leq C |\nu||(A)_c| \,.
\end{align*}
\end{step}
\end{proof}
\BBB
The next Lemma shows that our energy satisfies a genearlized coarea formula~\cite{ChaGiaLus,Visintin}.
\begin{lemma}\label{lemma:coarea} Let $u \colon \mathcal{L} \to \mathbb{R}$ and $A \subset \mathbb{R}^d$. Then
\begin{align}\label{eq:lemmacoarea}
E(u,A) = \int_{-\infty}^{+\infty} E(\chi_{\{u>t\}}(i),A)\,\mathrm{d}t\,.
\end{align}
\end{lemma}
\begin{proof} For $a,b \in \mathbb{R}$ we have
\begin{align*}
(a-b)^+ =\int_{-\infty}^{+\infty} (\chi_{\{a>t\}}-\chi_{\{b>t\}})^+\,\mathrm{d}t\,.
\end{align*}
Therefore,
\begin{align*}
E(u,A) = \sum_{i \in \mathcal{L}\cap A} \sum_{j \in \mathcal{L}} c_{i,j}(u(i)-u(j))^+ = \sum_{i \in \mathcal{L}\cap A} \sum_{j \in \mathcal{L}} c_{i,j}\int_{-\infty}^{+\infty} (\chi_{\{u(i)>t\}}-\chi_{\{u(j)>t\}})^+\,\mathrm{d}t\,.
\end{align*}
Now, by Fubini's Theorem (note that $c_{i,j} \geq 0$), we obtain
\begin{align*}
 \sum_{i \in \mathcal{L}\cap A} \sum_{j \in \mathcal{L}} c_{i,j}\int_{-\infty}^{+\infty} (\chi_{\{u(i)>t\}}-\chi_{\{u(j)>t\}})^+\,\mathrm{d}t&= \int_{-\infty}^{+\infty}\sum_{i \in \mathcal{L}\cap A} \sum_{j \in \mathcal{L}} c_{i,j}(\chi_{\{u(i)>t\}}-\chi_{\{u(j)>t\}})^+\,\mathrm{d}t \\&= \int_{-\infty}^{+\infty} E(\chi_{\{u>t\}},A)\,\mathrm{d}t
\end{align*}
and thus the claim.
\end{proof}
\EEE

\begin{lemma} \label{lemma:Requal01} Let $S>0$, \BBB $\delta>0$ and  $\nu \in \mathbb{S}^{d-1}$. \EEE Then
\begin{multline*}
\inf\left\{E(u,Q^\nu_S)\colon u \colon\mathcal{L} \to \mathbb{R}, u(i)=u_\nu(i) \text{ on } \mathcal{L}\setminus Q^\nu_{(1-\delta)S} \right\}\\=\inf\left\{E(u,Q^\nu_S)\colon u \colon\mathcal{L} \to \{0,1\}, u(i)=u_\nu(i) \text{ on } \mathcal{L}\setminus Q^\nu_{(1-\delta)S} \right\} \,.
\end{multline*}
\end{lemma}
\begin{proof} \BBB The infimum on the left hand side is taken over a larger class of admissible function, since here the image of the competitor $u$ is a subset of $\mathbb{R}$ and not just of $\{0,1\}$. Hence, one inequality is trivial.  The other inequality follows from a coarea formula satisfied by our energies. \EEE \\
\noindent\begin{step}{1} (Proof of '$\leq$') 
This inequality is clear, since the infimum on the left hand side is taken over a larger class of functions.
\end{step}\\
\noindent\begin{step}{2} (Proof of '$\geq$') \BBB
 Let us take $u \colon \mathcal{L} \to \mathbb{R}$ such that $u=u_\nu$ on  $\mathcal{L} \setminus Q^\nu_{(1-\delta)S}$ and denote by $u^s= \chi_{u>s}$.
Then, using Lemma \ref{lemma:coarea}, there exists $t\in (0,1)$ such that 
\begin{align*}
E(u^t,Q^\nu_S)\leq \int_0^1 E(u^s,Q^\nu_S)\,\mathrm{d}s\leq \int_{-\infty}^{+\infty} E(u^s,Q^\nu_S)\,\mathrm{d}s =E(u,Q^\nu_S) \,.
\end{align*}
Noting that $u^t(i) \in \{0,1\}$ for all $i \in \mathcal{L}$  and $u^t=u_\nu $ on $\mathcal{L}\setminus Q^\nu_S$, this concludes Step~2. \EEE
\end{step}
\end{proof}
Let $\phi \colon \mathbb{R}^d \to [0,+\infty]$ be defined by
\begin{align}\label{def:phi}
\phi(\nu) = \lim_{\delta\to 0}\lim_{S\to+\infty}\frac{1}{S^d}\inf\left\{E(u,Q_{S}) \colon u\colon \mathcal{L}\to \mathbb{R}, u(i) =\langle\nu,i\rangle \text{ on } \mathcal{L} \setminus Q_{(1-\delta)
S}\right\}\,.
\end{align}

$\phi_{\mathrm{per}} \colon \mathbb{R}^d \to [0,+\infty]$ is defined by
\begin{align}\label{def:phiper}
\phi_{\mathrm{per}}(\nu) = \BBB\liminf_{k\to+\infty}\EEE\frac{1}{(kT)^d}\inf\left\{E(u,Q_{kT}) \colon u\colon \mathcal{L}\to \mathbb{R}, u(\cdot) -\langle\nu,\cdot\rangle \BBB \in \mathcal{A}_\mathrm{per}(Q_{kT};\mathbb{R})\EEE\right\}\,.
\end{align}

\BBB The next lemma shows that $\phi_{\mathrm{per}}$ can be calculated via a finite cell formula. Additionally, it shows that the liminf in the definition of $\eqref{def:phiper}$ is actually a limit.
\EEE

\begin{lemma} \label{lemma:Treduction} Let \BBB $\nu \in \mathbb{S}^{d-1}$. \EEE For all $k \in \mathbb{N}$ there holds
\begin{align}\label{ineq:lemmaTreduction}
\begin{split}
\frac{1}{(kT)^d}&\inf\left\{E(u,Q_{kT}) \colon  u \colon \mathcal{L} \to \mathbb{R}, u(\cdot) -\langle\nu, \cdot\rangle \BBB  \in \mathcal{A}_{\mathrm{per}}(Q_{kT};\mathbb{R}) \EEE\right\} \\= \frac{1}{T^d}&\inf\left\{E(u,Q_{T}) \colon  u \colon \mathcal{L} \to \mathbb{R}, u(\cdot) -\langle\nu, \cdot\rangle  \BBB  \in \mathcal{A}_{\mathrm{per}}(Q_{T};\mathbb{R}) \EEE\right\} \, .
\end{split}
\end{align}
\BBB In particular,
\begin{align}\label{eq:lemmaTreduction}
\phi_{\rm per}(\nu) = \frac{1}{T^d}&\inf\left\{E(u,Q_{T}) \colon u \colon \mathcal{L} \to \mathbb{R},  u(\cdot) -\langle\nu, \cdot\rangle \BBB  \in \mathcal{A}_{\mathrm{per}}(Q_T;\mathbb{R})\EEE\right\}\,
\end{align}
\BBB and the $\liminf$ in \eqref{def:phiper} is actually
  a limit since the sequence, in fact, does not depend on $k$. \EEE
\end{lemma}
\BBB
The following example shows that, without any further assumption on $c_{i,j}$, the minimum in
\begin{align}\label{eq:inf}
\inf\left\{E(u,Q_{T}) \colon  u \colon \mathcal{L} \to \mathbb{R}, u(\cdot) -\langle\nu, \cdot\rangle    \in \mathcal{A}_{\mathrm{per}}(Q_{T};\mathbb{R}) \right\} 
\end{align}
is not achieved by any $u \colon \mathcal{L} \to \mathbb{R}$.
\begin{example}\label{ex:existence}  Let $\mathcal{L}=\mathbb{Z}$ and let $c_{i,j}$ be $2$-periodic. We set
\begin{align*}
c_{i,j}= \begin{cases} 2^{-k} &\text{if } j-i =2k+1, k \in \mathbb{N}, i \text{ even},\\
0 &\text{otherwise.}
\end{cases}
\end{align*}
Let $s\in \mathbb{R}$, $u(0)=0$ and $u(1)=s$ be $2$-periodic. Then, for $\nu=-1$ we have
\begin{align*}
E(u,Q_2) &= \sum_{k=0}^\infty c_{0,2k+1}(u(0)-u(2k+1)+2k+1)^+ =\sum_{k=0}^\infty 2^{-k}(2k+1-s)^+\,.
\end{align*}
We have that 
\begin{align*}
\inf_{u \in \mathcal{A}_{\mathrm{per}}(Q_2;\mathbb{R})} E(u,Q_2)=0\,.
\end{align*}
However, clearly, for all $u \in \mathcal{A}_{\mathrm{per}}(Q_2;\mathbb{R})$, $E(u,Q_2) >0$. In order to ensure the existence in \eqref{eq:inf} and the other minimum problems a coercivity condition might be: For any $i,j \in \mathcal{L}$ there exists a path $\gamma=(i_0,\ldots,i_N)$ such that $i_0=i$, $i_N=j$ and such that $c_{i_k,i_{k+1}} >0$. By considering the directed graph $G=(\mathcal{L}, \mathcal{E})$, where $\mathcal{E}:=\{(i,j) \in \mathcal{L} \times \mathcal{L}\colon \text{ there exists } z \in \mathbb{Z}^d \text{ such that }   c_{i,j+Tz}>0\}$ this condition ensures that for two vertices $i,j \in G$ there always exists a path of edges (in the infinite graph) with positive weights connecting them.
\end{example}
\EEE

\begin{proof}[Proof of Lemma \ref{lemma:Treduction}] We split the proof into two steps by first observing the (obvious) inequality that the right hand side in \eqref{ineq:lemmaTreduction} is less than or equal to the left hand side. Then, we prove the converse inequality by using a superposition argument.

\BBB
  \noindent\begin{step}{1}(Proof of '$\leq$')
    given $u$ with $u(\cdot)-\langle\nu,\cdot\rangle\in
    \mathcal{A}_{\mathrm{per}}(Q_T;\mathbb{R})$, then obviously
    $u(\cdot)-\langle\nu,\cdot\rangle\in
    \mathcal{A}_{\mathrm{per}}(Q_{kT};\mathbb{R})$ and
    \[
      E(u,Q_{kT}) = \frac{1}{k^d}E(u,Q_T)
    \]
    so this inequality is obvious.
  \end{step}

  \noindent\begin{step}{2}(Proof of '$\geq$')
    This is a standard convexity argument: given $u$ now with
    $u(\cdot)-\langle\nu,\cdot\rangle\in\mathcal{A}_{\mathrm{per}}(Q_{kT};\mathbb{R})$,
    then for $i\in Q_T$ we let:
    \[
      u_T(i) = \langle\nu,i\rangle
      + \frac{1}{k^d} \sum_{z \in \{0,\ldots,k-1\}^d} (u(i+Tz)-\langle\nu,
      i+Tz\rangle)\,.
    \]
    Then clearly by construction, $u_T\in \mathcal{A}_{\mathrm{per}}(Q_{T};\mathbb{R})$ and by convexity,
    \[
      E(u_T,Q_T)= \frac{1}{k^d} E(u_T,Q_{kT})\le \frac{1}{k^d}E(u,Q_{kT}), 
    \]
    which shows the lemma.
  \end{step}
\EEE

\end{proof}

The following lemma uses a standard cutoff-argument. However, due to the infinite range of interactions, \BBB the arguments for the case of finite range interactions need to be adapted\EEE.

\begin{lemma}\label{lemma: per=aff} Let $\nu \in \mathbb{R}^d$. Then:
  $\phi_{\rm per}(\nu) = \phi(\nu).$
\end{lemma}
\begin{proof} We first show $\phi_{\rm per}(\nu) \leq \phi(\nu)$ in Step~1, and then the reverse inequality. In order to do so, we modify competitors of the respective cell formulas in order to obtain a competitor for the other formula. 
Due to the one homogeneity of both functions, we may assume that $\nu \in \mathbb{S}^{d-1}$.

\noindent \begin{step}{1}(Proof of  '$\leq$') \BBB Due to Lemma \ref{lemma:Treduction}, the limit  in the definition of $\phi_{\rm per}(\nu)$ (resp. $\phi(\nu)$) exists. \EEE  Thus, we can assume without loss of generality that $S=kT$ for some $k \in \mathbb{N}$ with $k$ large. Let $\delta>0$, \BBB $\varepsilon>0$, \EEE and let $u_k^\delta \colon \mathcal{L} \to \mathbb{R}$ be such that $u_k^\delta(i) = \langle\nu,i\rangle$ on $ \mathcal{L} \setminus Q_{(1-\delta)kT}$ and \BBB
\begin{align}\label{eq:minukdelta}
\begin{split}
E(u_k^\delta,Q_{kT}) &\leq\inf\left\{E(u,Q_{kT}) \colon u\colon \mathcal{L}\to \mathbb{R}, u(i) =\langle\nu,i\rangle \text{ on } \mathcal{L} \setminus Q_{(1-\delta)
kT}\right\} +\varepsilon \,.
\end{split}
\end{align} \EEE
We assume that
\begin{align}\label{ineq:uinfty bound}
||u_k^\delta||_{L^\infty(Q_{(1+\delta )kT})} \leq 2kT\,.
\end{align}
If that were not true we perform the following construction with $\tilde{u}_k^\delta(i) = (u_k^\delta(i) \vee (-2kT))\wedge (2kT)$. Note that still $\tilde{u}_k^\delta(i) =\langle\nu,i\rangle$ on $Q_{(1+\delta)kT}$ for $\delta$ small enough.
We define $v_k^\delta \colon \mathcal{L} \to \mathbb{R}$ by setting
\BBB
  \begin{equation}\label{def:vkdelta}
    v_k^\delta(i)
    - \langle \nu,i\rangle
    = u_k^\delta(i_0)-\langle\nu,i_0\rangle \text{ if }
    i=i_0+kTz, i_0 \in Q_{kT}, z \in \mathbb{Z}^d
  \end{equation}
  so that $v_k^\delta(\cdot)-\langle \nu,\cdot\rangle\in \mathcal{A}_{\mathrm{per}}(Q_{kT};\mathbb{R})$. Then clearly,
  writing $i=i_0+kTz$ and $j=j_0+k'T_z$ as above:
\begin{align}\label{ineq:vkdeltalipschitz}
  |v_k^\delta(i)-v_k^\delta(j)| \le |u_k^\delta(i_0)-u_k^\delta(j_0)|+
  |i_0-j_0| + |i-j| 
  \leq C kT + |i-j|\,,
\end{align}
and
\begin{align}\label{eq:vkdeltaukdelta}
 v_k^\delta(i) = u_k^\delta(i) \text{ for } i \in Q_{(1+\delta)kT}\,, 
\end{align}
since $u_k^\delta(i)-\langle\nu,i\rangle=0$ for $i\not\in Q_{(1-\delta)kT}$. \EEE
Additionally,
\begin{align}\label{ineq:vkdeltainf}
\inf\left\{E(u,Q_{kT}) \colon u\colon \mathcal{L}\to \mathbb{R}, u(\cdot) -\langle\nu,\cdot\rangle \BBB \in \mathcal{A}_{\mathrm{per}}(Q_{kT};\mathbb{R})\EEE\right\} \leq E(v_k^\delta, Q_{kT})\,.
\end{align}
We are finished with Step 1 if we prove
\begin{align}\label{ineq:comparison energy ukdeltavkdelta}
E(v_k^\delta,Q_{kT}) \leq E(u_k^\delta,Q_{kT}) + \BBB \frac{C_k^\delta}{\delta}(kT)^d\,,\EEE
\end{align}
\BBB where $C_k^\delta \to 0$ as $k\to +\infty$. \EEE In fact, using \eqref{eq:minukdelta}, \eqref{ineq:vkdeltainf}, \eqref{ineq:comparison energy ukdeltavkdelta}, dividing by $(kT)^d$, letting $k\to +\infty$, and then $\delta \to 0$, we obtain the claim \BBB by noting that $\varepsilon>0$ is chosen arbitrarily\EEE. Let us prove \eqref{ineq:comparison energy ukdeltavkdelta}. We have, using \eqref{eq:vkdeltaukdelta}, 
\begin{align*}
E(v_k^\delta,Q_{kT}) &= \sum_{ i \in \mathcal{L}\cap Q_{kT}} \sum_{j \in \mathcal{L}} c_{i,j}(v_k^\delta(i)-v_k^\delta(j))^+ \\&= \sum_{i \in \mathcal{L}\cap Q_{kT}} \sum_{j \in \mathcal{L} \cap Q_{(1+\delta)kT}}\!\!\! c_{i,j}(v_k^\delta(i)-v_k^\delta(j))^+ +  \sum_{i \in \mathcal{L}\cap Q_{kT}} \sum_{j \in \mathcal{L} \setminus Q_{(1+\delta)kT}} \!\!\!c_{i,j}(v_k^\delta(i)-v_k^\delta(j))^+  \\& \leq E(u_k^\delta, Q_{kT}) +\sum_{i \in \mathcal{L}\cap Q_{kT}} \underset{|i-j|\geq \delta kT/2}{\sum_{j \in \mathcal{L} }} c_{i,j}|v_k^\delta(i)-v_k^\delta(j)|\,.
\end{align*}
Hence, in order to show \eqref{ineq:comparison energy ukdeltavkdelta}, it remains to prove
\begin{align}\label{ineq:remaindervkdelta}
\sum_{i \in \BBB \mathcal{L} \cap Q_{kT}\EEE} \underset{|i-j|\geq \delta kT/2}{\sum_{j \in \mathcal{L} }} c_{i,j}|v_k^\delta(i)-v_k^\delta(j)|  \leq \frac{C_k^\delta}{\delta} (kT)^d\,,
\end{align}
\BBB where $C_k^\delta \to 0$ as $k\to+\infty$. \EEE
Using \eqref{ineq:vkdeltalipschitz}, {\rm (H2)}, and Lemma \ref{lemma:elementaryproperties}(v), we have
\begin{align*}
\sum_{i \in \mathcal{L}\cap Q_{kT}} \underset{|i-j|\geq \delta kT/2}{\sum_{j \in \mathcal{L} }} c_{i,j}|v_k^\delta(i)-v_k^\delta(j)| &\leq  \sum_{i \in \mathcal{L} \cap Q_{kT}} \underset{|i-j|\geq \delta kT/2}{\sum_{j \in \mathcal{L} }} c_{i,j}( CkT +|i-j|) \\&\leq \left(\frac{C}{\delta}+1\right)\sum_{i \in \mathcal{L}\cap Q_{kT}} \underset{|i-j|\geq \delta kT/2}{\sum_{j \in \mathcal{L} }} c_{i,j}|i-j| \\&\leq \frac{C}{\delta}  \#\left(\mathcal{L}\cap Q_{kT}\right)\max_{i \in \mathcal{L}} \underset{|i-j|\geq \delta kT/2}{\sum_{j \in \mathcal{L} }} c_{i,j}|i-j|  \leq \frac{C_k^\delta}{\delta}(kT)^d\,,
\end{align*}
where $C_k^\delta\to 0$ as $k \to +\infty$.
 This  yields \eqref{ineq:remaindervkdelta} and therefore the claim of Step 1.
\end{step}\\
\noindent\begin{step}{2}(Proof of '$\geq$') \BBB Let $\varepsilon>0$ and \EEE $u \colon \mathcal{L} \to  \mathbb{R}$ be such that \BBB $u(\cdot) - \langle\nu, \cdot \rangle \in \mathcal{A}_\mathrm{per}(Q_{T};\mathbb{R})$ \EEE and 
\begin{align*}
E(u,Q_{T}) \leq  \inf\left\{E(u,Q_{T}) \colon u \colon \mathcal{L} \to \mathbb{R}, u(\cdot) -\langle\nu, \cdot\rangle \BBB \in \mathcal{A}_{\mathrm{per}}(Q_{T};\mathbb{R})\EEE \right\} +\varepsilon\,.
\end{align*}
\BBB Fix $\delta>0$ and $S \in \mathbb{N}$ such that $S =kT$ for some $k\in \mathbb{N},\, k \gg 1$ and \BBB $\delta S \gg 1$\EEE. Since $u(\cdot) - \langle\nu, \cdot \rangle \in \mathcal{A}_{\mathrm{per}}(Q_{T};\mathbb{R})$, we have \EEE
\begin{align*}
E(u,Q_{T}(x_0))=E(u,Q_{T}) \text{ for all } x_0 \in T\mathbb{Z}^d
\end{align*}
and therefore
\begin{align}\label{eq:periodicity energy}
E(u,Q_{S})= \frac{S^d}{T^d}E(u,Q_{T})\,.
\end{align}
There exists a constant $C >0$ (we omit the dependence on $T$) such that, due to the fact that $u(\cdot)-\langle\nu,\cdot\rangle\BBB\in \mathcal{A}_{\mathrm{per}}(Q_{T};\mathbb{R})$,\EEE  there holds
\begin{align}\label{ineq:estimate kTcell}
\max_{i \in \mathcal{L}}|u(i) -\langle\nu,i\rangle | = \max_{i \in \mathcal{L}\cap Q_{T}} |u(i) -\langle\nu,i\rangle | \leq C_\varepsilon \,.
\end{align}
Let \BBB $\zeta_S \in C_c^\infty(\mathbb{R}^d;[0,1])$ \EEE be a cut-off function such that
\begin{align*}
\BBB\zeta_S(x) =1 \text{ for } x \in Q_{(1-3\delta)S}\,,\quad \mathrm{supp}\,\zeta_S(x) \subset  Q_{(1-2\delta)S}\,,\quad \text{ and }  ||\nabla \zeta_S ||_\infty\leq \frac{C}{\delta S}\,.\EEE
\end{align*}
Define $u_S \colon \mathcal{L} \to \mathbb{R}$ by
\begin{align*}
u_S(i) = \BBB \zeta_S(i)u(i) + (1-\zeta_S(i))\langle\nu,i\rangle\,.\EEE
\end{align*}
Then, $u_S(i) = \langle \nu,i\rangle$ for $i \in \mathcal{L}\setminus Q_{(1-\delta)S}$ and therefore
\begin{align}\label{ineq:uSmin}
 \inf\left\{E(u,Q_{S}) \colon u \colon \mathcal{L} \to \mathbb{R}, u(i) =\langle\nu, i\rangle  \text{ on }  \mathcal{L}\setminus Q_{(1-\delta)S} \right\} \leq E(u_S,Q_S)\,.
\end{align}
 For all $i,j \in \mathcal{L}$ there holds
\begin{align*}
\BBB u_S(i)-u_S(j) = \zeta_S(i)\left(u(i)-u(j)\right) + (1-\zeta_S(i))\langle \nu,i-j\rangle + (\zeta_S(i)-\zeta_S(j))(u(j)-\langle\nu,j\rangle)\,, \EEE
\end{align*}
which, together with \eqref{ineq:estimate kTcell},  implies for all $i,j \in \mathcal{L}$
\begin{align}\label{ineq:differencebound}
\begin{split}
(u_S(i)-u_S(j))^+ &\leq (u(i)-u(j))^+ + |i-j| + \BBB \frac{C}{\delta S}|u(j)-\langle\nu,j\rangle| |i-j|\EEE \\&\leq (u(i)-u(j))^+ +C |i-j| \,,
\end{split}
\end{align}
\BBB where we assume that $S\delta  \geq C_\varepsilon$ (we will first send  $k$ to $+\infty$, then $\delta$ to $0$, and finally $\varepsilon$ to $0$). \EEE
For all $i,j \in Q_{(1-3\delta)S}$ we have
\begin{align}\label{ineq:differenceboundinside}
(u_S(i)-u_S(j))^+ = (u(i)-u(j))^+\,.
\end{align}
Using \eqref{eq:periodicity energy}, \eqref{ineq:differencebound}, and \eqref{ineq:differenceboundinside}, we obtain 
\begin{align}\label{ineq:energyuSQS}
\begin{split}
E(u_S,Q_S) &\leq \sum_{i \in\mathcal{L} \cap Q_{S}} \sum_{j \in \mathcal{L}} c_{i,j} (u(i)-u(j))^+ +   C\sum_{i \in \mathcal{L}\cap Q_{(1-6\delta)S}}  \sum_{j \in \mathcal{L} \setminus Q_{(1-3\delta)S}} c_{i,j} |i-j| \\&\quad+C\sum_{i \in \mathcal{L} \cap Q_{S} \setminus Q_{(1-6\delta)S}} \sum_{j \in \mathcal{L}} c_{i,j}|i-j| \\&= \frac{S^d}{(kT)^d} E(u_k,Q_{kT})+  C \sum_{i \in \mathcal{L}\cap Q_{(1-6\delta)S}}  \sum_{j \in \mathcal{L} \setminus Q_{(1-3\delta)S}} c_{i,j} |i-j|\\&\quad+C\sum_{i \in \mathcal{L} \cap Q_{S} \setminus Q_{(1-6\delta)S}} \sum_{j \in \mathcal{L}} c_{i,j}|i-j|\,.
\end{split}
\end{align} 
 We show that 
\begin{align}\label{ineq:firstStep2per}
\sum_{i \in \mathcal{L}\cap Q_{(1-6\delta)S}}  \sum_{j \in \mathcal{L} \setminus Q_{(1-3\delta)S}} c_{i,j} |i-j| \leq \BBB C_S^\delta \EEE S^d\,,
\end{align}
and
\begin{align}\label{ineq:secondStep2per}
\sum_{i \in \mathcal{L} \cap Q_{S} \setminus Q_{(1-6\delta)S}} \sum_{j \in \mathcal{L}} c_{i,j}|i-j|  \leq C \delta S^{d}\,,
\end{align}
where $\BBB C_S^\delta \to 0$ \EEE as $S\to +\infty$ and $C>0$ is a universal constant. \EEE
Note that, due to Lemma \ref{lemma:Treduction}, \BBB since $\varepsilon>0$ is chosen arbitrary, \EEE from \eqref{ineq:firstStep2per} and \eqref{ineq:secondStep2per}  we obtain the claim of Step 2 by using \eqref{ineq:uSmin}, \eqref{ineq:energyuSQS}, dividing by $S^d$, $k \to +\infty$ and then $\delta \to 0$.

We first prove  \eqref{ineq:firstStep2per}. Note that, for $S$ big enough, due to {\rm (H2)} and Lemma \ref{lemma:elementaryproperties}(v), we have
\begin{align*}
\sum_{i \in \mathcal{L}\cap Q_{(1-6\delta)S}}  \sum_{j \in \mathcal{L} \setminus Q_{(1-3\delta)S}} c_{i,j} |i-j| & \leq \sum_{i \in \mathcal{L}\cap Q_{(1-6\delta)S}}  \underset{|i-j|\geq \delta S}{\sum_{j \in \mathcal{L} }}c_{i,j} |i-j| \\&\leq \#(\mathcal{L}\cap Q_S) \max_{i \in \mathcal{L}} \underset{|i-j|\geq \delta S}{\sum_{j \in \mathcal{L}}}c_{i,j}|i-j| \\&\leq \BBB C_S^\delta \EEE  S^d \,,
\end{align*}
where \BBB $C_S^\delta \to 0$ \EEE as $S \to +\infty$. Next, we show \eqref{ineq:secondStep2per}. Using {\rm (H2)}, and Lemma \ref{lemma:elementaryproperties}(v), we obtain
\begin{align*}
\sum_{i \in \mathcal{L} \cap Q_{S} \setminus Q_{(1-6\delta)S}} \sum_{j \in \mathcal{L}} c_{i,j}|i-j| &\leq \sum_{i \in \mathcal{L} \cap Q_{S} \setminus Q_{(1-6\delta)S}} \sum_{j \in \mathcal{L}} c_{i,j} |i-j|\\&\leq   \# (\mathcal{L}\cap Q_{S} \setminus Q_{(1-6\delta)S})\max_{i \in \mathcal{L}} \sum_{j \in \mathcal{L}} c_{i,j}|i-j|\\&\leq C\delta S^d\,.
\end{align*}
This is \eqref{ineq:secondStep2per} and hence the claim of Step 2.
\end{step}
\end{proof}

Let $\psi \colon \mathbb{R}^d \to [0,+\infty]$ be defined as the positively homogeneous function of degree one that for $\nu \in \mathbb{S}^{d-1}$ is defined by
\begin{align}\label{def:psi}
\psi(\nu) = \lim_{\delta \to 0} \lim_{S\to +\infty} \frac{1}{S^d}\inf\left\{E(u,Q^\nu_S) \colon u \colon \mathcal{L} \to \mathbb{R}, u(i)=\langle \nu,i\rangle \text{ on } \mathcal{L}\setminus Q^\nu_{(1-\delta)S}\right\} \,.
\end{align}

\BBB The function $\psi$ differs from the function $\phi$ in the domain where one calculates the energy. For the function $\phi$ we take the coordinate cube $Q_T$ whereas for $\nu$ we take the cube $Q^\nu_T$. 
\begin{remark}\label{rem: existence of the limits} The existence of the limits in \eqref{def:phi} and \eqref{def:psi} can be deduced from standard subadditivity arguments, see e.g.~\cite[Proposition 4.2]{AliCic}.
\end{remark}

\EEE

\begin{lemma}\label{lemma:propertiesofpsi} $\psi \colon \mathbb{R}^d \to [0,+\infty]$ satisfies the following properties:
\begin{enumerate}[label={\rm(\roman*)}]
\item There exists $C>0$ such that $\psi(\nu) \leq C|\nu|$ for all $\nu \in \mathbb{R}^d$\, ,
\item $\psi$ is a continuous function.
\end{enumerate}
\end{lemma}

\begin{proof} We divide the proof into two steps. We first prove {\rm (i)} and then {\rm (ii)}. Throughout the proofs let $1\ll S$.\\ \noindent\begin{step}{1}(Proof of {\rm (i)}) \BBB Let $\nu \in \mathbb{S}^{d-1}$; \EEE it suffices to prove
\begin{align*}
\psi(\nu) \leq C\,.
\end{align*}
The general case then follows by one-homogeneity.
 In order to prove {\rm (i)} we insert $u(i)=\langle\nu,i\rangle$ for all $i \in \mathcal{L}$ as a competitor in the cell formula. Using Lemma \ref{lemma:elementaryproperties}(i), we then have
\begin{align*}
E(u,Q^\nu_S) =E(\langle\nu,\cdot\rangle,Q^\nu_S) \leq C|\nu| |(Q_S)_c| \leq CS^d\,.
\end{align*}
Dividing by $S^d$ and letting $S \to +\infty$ yields the claim.
\end{step}\\
\noindent\begin{step}{2}(Proof of {\rm (ii)}) Due to the one-homogeneity, it suffices to consider the case where $\nu_1,\nu_2 \in \mathbb{S}^{d-1}$. Let $\eta>0$ and $\nu_1,\nu_2 \in \mathbb{S}^{d-1}$ be such that $|\nu_1-\nu_2| \leq \eta$. Our goal is to prove that there exists $C>0$ independent of $\nu_1$ and $\nu_2$ such that
\begin{align}\label{ineq:continuity}
|\psi(\nu_1) -\psi(\nu_2)| \leq C\eta\,.
\end{align}
We only prove 
\begin{align}\label{ineq:contonesided}
\psi(\nu_1) -\psi(\nu_2) \leq C\eta\, ,
\end{align}
since then \eqref{ineq:continuity} follows by exchanging $\nu_1$ and $\nu_2$ in \eqref{ineq:contonesided}. To this end let $\delta>0$ small enough, $S>0$ big enough, $u_1 \colon \mathcal{L} \to \mathbb{R}$ be such that $u_1(i) = \langle\nu_1,i\rangle$ on $\mathcal{L} \setminus Q^\nu_{(1-\delta)S}$  and
\begin{align}\label{ineq:psinu1}
\frac{1}{S^d}E(u_1,Q^{\nu_1}_S) \leq \phi(\nu_1) +\eta\,.
\end{align}
We assume that 
\begin{align}\label{ineq:u1infinityinsidenu1}
||u_1||_{L^\infty(Q^{\nu_1}_S)} \leq S\,.
\end{align}
If this were not the case, we consider 
\begin{align*}
\tilde{u}_1(i) = \begin{cases} (u_1(i) \wedge S) \vee (-S) &i \in Q^{\nu_1}_{2S}\,, \\
u_1(i) &\text{otherwise.}
\end{cases}
\end{align*}
Note that for $i,j \in Q^{\nu_1}_{2S}$, due to truncation, $(\tilde{u}_1(i)-\tilde{u}_1(j))^+\leq (u_1(i)-u_1(j))^+$, whereas in general there holds $ |\tilde{u}_1(i)-\tilde{u}_1(j)| \leq CS +|i-j|$. From this, using Lemma \ref{lemma:elementaryproperties}(v) and {\rm (H2)}, we deduce
\begin{align*}
E(\tilde{u}_1,Q^{\nu_1}_S) &= \sum_{i \in \mathcal{L}\cap Q^{\nu_1}_S} \sum_{j \in \mathcal{L}\cap Q^{\nu_1}_{2S}} c_{i,j}(\tilde{u}_1(i)-\tilde{u}_1(j))^+ + \sum_{i \in \mathcal{L}\cap Q^{\nu_1}_S} \sum_{j \in \mathcal{L}\setminus Q^{\nu_1}_{2S}} c_{i,j}(\tilde{u}_1(i)-\tilde{u}_1(j))^+\\&\leq \sum_{i \in \mathcal{L}\cap Q^{\nu_1}_S} \sum_{j \in \mathcal{L}\cap Q^{\nu_1}_{2S}} c_{i,j}(u_1(i)-u_1(j))^+ + C\sum_{i \in \mathcal{L}\cap Q^{\nu_1}_S} \underset{|i-j|\geq S/2}{\sum_{j \in \mathcal{L}}} c_{i,j}|i-j| \\&\leq E(u_1,Q^{\nu_1}_S) + C \#(\mathcal{L}\cap Q^{\nu_1}_{S}) \max_{i \in \mathcal{L}} \underset{|i-j|\geq S/2}{\sum_{j \in \mathcal{L}}} c_{i,j}|i-j|  \leq E(u_1,Q^{\nu_1}_S) + C_S S^d\,,
\end{align*}
where $C_S \to 0$ as $S \to \infty$. In particular $C_S \leq \eta$ for $S$ big enough. Hence, we can assume \eqref{ineq:u1infinityinsidenu1}.
 There exists $C>0$ such that for $\tilde{S}=(1+C\eta)S$ there holds $Q^{\nu_2}_{(1-\delta)\tilde{S}} \supset Q^{\nu_1}_{(1+\delta)S}$. We now define $u_2 \colon \mathcal{L} \to \mathbb{R}$ by 
\begin{align*}
u_2(i) = \langle\nu_2-\nu_1,i\rangle +u_1(i)\,.
\end{align*}
First, note that $u_2(i) = \langle\nu_2,i\rangle$ for all $i \in \mathcal{L} \setminus Q^{\nu_2}_{(1-\delta)\tilde{S}}$ and therefore
\begin{align}\label{ineq:infadmissible}
\inf\left\{E(u,Q^{\nu_2}_{\tilde{S}}) \colon u \colon \mathcal{L} \to \mathbb{R}, u(i) =\langle\nu_2,i\rangle \text{ on } \mathcal{L} \setminus Q^{\nu_2}_{(1-\delta)\tilde{S}} \right\} \leq E(u_2,Q^{\nu_2}_{\tilde{S}})\,.
\end{align}
We claim that  \BBB
\begin{align}\label{ineq:u2u1error}
E(u_2,Q^{\nu_2}_{\tilde{S}}) \leq E(u_1Q^{\nu_1}_S) +  \frac{C_S^\delta}{\delta} S^d +C\eta S^d +C\delta S^d\,,
\end{align}
where $C_S^\delta \to 0$ as $S \to +\infty$. \EEE We postpone the proof of \eqref{ineq:u2u1error} and show first how it implies \eqref{ineq:contonesided}. Dividing \eqref{ineq:u2u1error} by $\tilde{S}^d$, letting $\tilde{S}$ (therefore also $S$) tend to $+\infty$, $\delta \to 0$, and using  \eqref{ineq:infadmissible} as well as \eqref{ineq:psinu1}, we get
\begin{align*}
\phi(\nu_2) \leq \phi(\nu_1) +C\eta \leq \phi(\nu_1) + C\eta\,.
\end{align*}
This is \eqref{ineq:contonesided}. We now prove \eqref{ineq:u2u1error}. Due to Lemma \ref{lemma:elementaryproperties}(ii), there holds
\begin{align}\label{eq:threetermsplit}
E(u_2,Q^{\nu_2}_{\tilde{S}})\leq  E(u_1,Q^{\nu_2}_{\tilde{S}})+E(\langle\nu_2-\nu_1,\cdot\rangle,Q^{\nu_1}_{\tilde{S}})\,.
\end{align}
Now, due to Lemma \ref{lemma:elementaryproperties}(i) and the fact that $\tilde{S}\leq 2S$, there holds
\begin{align}\label{ineq:nudiffbound}
E(\langle\nu_2-\nu_1,\cdot\rangle,Q^{\nu_1}_{\tilde{S}}) \leq C|\nu_2-\nu_1| S^d\leq C\eta S^d\,.
\end{align}
Next, we prove  \BBB
\begin{align}\label{ineq:u1insidenu2}
E(u_1,Q^{\nu_2}_{\tilde{S}}) \leq E(u_1,Q^{\nu_1}_{S}) + C\delta S^d + \frac{C_S^\delta}{\delta} S^d\,,
\end{align}
where $C_S^\delta \to 0$ as $S \to +\infty$. \EEE We use Lemma \ref{lemma:elementaryproperties}(iv), to obtain
\begin{align*}
E(u_1,Q^{\nu_2}_{\tilde{S}})= E(u_1,Q^{\nu_1}_{S}) + E(u_1,Q^{\nu_2}_{\tilde{S}}\setminus Q^{\nu_1}_S)\,.
\end{align*}
In order to prove \eqref{ineq:u1insidenu2} it suffices to prove 
\begin{align}\label{ineq:u1insidenu2second}
E(u_1,Q^{\nu_2}_{\tilde{S}}\setminus Q^{\nu_1}_S) \leq C\eta S^d + \frac{C_S^\delta}{\delta} S^d\,,
\end{align}
where $C_{S}^\delta \to 0$ as $S\to +\infty$.
To see this we write
\begin{align}\label{eq:u1split}
\begin{split}
E(u_1,Q^{\nu_2}_{\tilde{S}}\setminus Q^{\nu_1}_S)  &= \sum_{i \in \mathcal{L} \cap Q^{\nu_2}_{\tilde{S}}\setminus Q^{\nu_1}_S} \sum_{j \in \mathcal{L}\cap Q^{\nu_1}_{(1-\delta)S}} c_{i,j}(u_1(i)-u_1(j))^+ \\&\quad+ \sum_{i \in \mathcal{L} \cap Q^{\nu_2}_{\tilde{S}}\setminus Q^{\nu_1}_S} \sum_{j \in \mathcal{L}\setminus Q^{\nu_1}_{(1-\delta)S}} c_{i,j}(u_1(i)-u_1(j))^+\,.
\end{split}
\end{align}
To estimate the first term, note that due to \eqref{ineq:u1infinityinsidenu1}, we have $|u_1(i)-u_1(j)| \leq CS +|i-j|$, and therefore, up to changing $C$, using {\rm (H2)}, and Lemma \ref{lemma:elementaryproperties}(iv),  we get
\begin{align}\label{ineq:u1splitfirstterm}
\sum_{i \in \mathcal{L} \cap Q^{\nu_2}_{\tilde{S}}\setminus Q^{\nu_1}_S} \sum_{j \in \mathcal{L}\cap Q^{\nu_1}_{(1-\delta)S}} c_{i,j}(u_1(i)-u_1(j))^+ &\leq \frac{C}{\delta}\sum_{i \in \mathcal{L} \cap Q^{\nu_2}_{\tilde{S}}\setminus Q^{\nu_1}_S} \underset{|i-j|\geq \delta S/2}{\sum_{j \in \mathcal{L}}} c_{i,j}|i-j|\\&\leq \frac{C}{\delta}\#(\mathcal{L}\cap Q^{\nu_2}_{\tilde{S}_2}) \max_{i \in \mathcal{L}} \underset{|i-j|\geq \delta S/2}{\sum_{j \in \mathcal{L}}} c_{i,j}|i-j| \leq \frac{C_S^\delta}{\delta} S^d\,, \nonumber
\end{align} 
\BBB where $C_S^\delta \to 0$ as $S\to +\infty$. \EEE
To estimate the first term, we use the fact that $u_1(i) = \langle\nu_1,i\rangle$ on $\mathcal{L}\setminus Q^{\nu_1}_{(1-\delta)S}$,  and Lemma \ref{lemma:elementaryproperties} (i), to obtain
\begin{align*}
\sum_{i \in \mathcal{L} \cap Q^{\nu_2}_{\tilde{S}}\setminus Q^{\nu_1}_S} \sum_{j \in \mathcal{L}\setminus Q^{\nu_1}_{(1-\delta)S}} \!\!\!\!\!c_{i,j}(u_1(i)-u_1(j))^+ \leq E(\langle\nu_1,\cdot\rangle,Q^{\nu_2}_{\tilde{S}}\setminus Q^{\nu_1}_S) \leq C|\nu_1| |(Q^{\nu_2}_{\tilde{S}}\setminus Q^{\nu_1}_S)_c|\leq C\eta S^d\,.
\end{align*}
This together with \eqref{eq:u1split} and \eqref{ineq:u1splitfirstterm} implies \eqref{ineq:u1insidenu2second} which in turn, together with \eqref{eq:threetermsplit} and \eqref{ineq:nudiffbound} implies \eqref{ineq:u2u1error} and therefore the conclusion of Step 2.
\end{step}
\end{proof}

\begin{lemma}\label{lemma:propertiesofphi} $\phi \colon \mathbb{R}^d \to [0,+\infty]$ satisfies the following properties:
\begin{enumerate}[label={\rm(\roman*)}]
\item There exists $C>0$ such that $\phi(\nu) \leq C|\nu|$ for all $\nu \in \mathbb{R}^d$\, ,
\item $\phi$ is a positively homogeneous function of degree one,
\item $\phi$ is a \BBB convex function. In particular, $\phi$ is Lipschitz continuous. \EEE
\end{enumerate}
\end{lemma}

\begin{proof} We divide the proof into two steps. Throughout the proofs let $1\ll S$.\\
\noindent\begin{step}{1}(Proof of {\rm (i)} and {\rm (ii})) 
 In order to prove {\rm (i)} we insert $u(i)=\langle\nu,i\rangle$ for all $i \in \mathcal{L}$ as a competitor in the cell formula. Using Lemma \ref{lemma:elementaryproperties}(i), we then have
\begin{align*}
E(u,Q^\nu_S) = E(\langle\nu,\cdot\rangle,Q^\nu_S) \leq C|\nu| |(Q_S)_c| \leq CS^d\,.
\end{align*}
Dividing by $S^d$ and letting $S \to +\infty$ yields the claim.
{\rm (ii)} follows by using Lemma \eqref{lemma:elementaryproperties}(ii) to obtain $E(\lambda u, Q_S) = \lambda E(u,Q_S)$ for all $\lambda >0$ and by noting that, given $\nu \in \mathbb{R}^d$, if $u \colon \mathcal{L} \to \mathbb{R}$ satisfies $u(i) = \langle\nu,i\rangle $ on $\mathcal{L}\setminus Q_{(1-\delta)S}$, then $\lambda u(i) = \langle\lambda\nu,i\rangle $ on $\mathcal{L}\setminus Q_{(1-\delta)S}$. Employing this in \eqref{def:phi} it is easy to see that $\phi$ is a \BBB positively homogeneous function of degree one\EEE.
\end{step}\\
\noindent\begin{step}{2}(Proof of {\rm (iii)}) \BBB We show that for every $S>0$ and $\delta>0$, $S\delta\gg 1$ the function  $\phi_S^\delta \colon \mathbb{R}^d \to [0,+\infty]$ given by
\begin{align}
\phi_S^\delta(\nu) := \frac{1}{S^d}\inf\left\{E(u,Q_{S}) \colon u\colon \mathcal{L}\to \mathbb{R}, u(i) =\langle\nu,i\rangle \text{ on } \mathcal{L} \setminus Q_{(1-\delta)S}\right\} 
\end{align}
is a convex function. Note that 
\begin{align*}
\phi(\nu) = \lim_{\delta \to 0} \lim_{S\to +\infty} \phi_S^\delta(\nu) \text{ for all } \nu \in \mathbb{R}^d\,.
\end{align*}
Thus, the convexity for $\phi_S^\delta$ also implies the convexity of $\phi$. This together with {\rm (i)} and {\rm (ii)} implies that $\phi$ is also Lipschitz continuous. Now we prove that $\phi_S^\delta$ is a convex function. Given $\lambda \in [0,1]$, $\nu_1, \nu_2 \in \mathbb{R}^d$, let $\varepsilon>0$ $u_1 \colon \mathcal{L} \to \mathbb{R}$, $u_2 \colon \mathcal{L}\to \mathbb{R}$ be such that $u_k(i)= \langle\nu_k,i\rangle$ on $\mathcal{L}\setminus Q_{(1-\delta)S}$  and 
\begin{align*}
E(u_k,Q_S)\leq  \phi_S^\delta(\nu_k)+\frac{1}{2}\varepsilon \text{ for } k=1,2\,.
\end{align*}
We have that $u(i):= \lambda u_1(i) +(1-\lambda) u_2$ is admissible for $\phi_S^\delta(\lambda\nu_1+(1-\lambda)\nu_2)$ and by Lemma \ref{lemma:elementaryproperties}(ii) we obtain
\begin{align*}
\phi_S^\delta(\lambda\nu_1+(1-\lambda)\nu_2) &\leq E(u,Q_S) = E(\lambda u_1 +(1-\lambda)u_2,Q_S)\leq \lambda E(u_1,Q_S) +(1-\lambda) E(u_2,Q_S) \\&\leq \lambda \phi_S^\delta(\nu_1)+ (1-\lambda)\phi_S^\delta(\nu_2) +\varepsilon\,.
\end{align*}
Since $\varepsilon>0$ is arbitrary. This yields the claim.
 \EEE
\end{step}
\end{proof}

The next Lemma shows that the asymptotic cell-formula  describing the surface energy density is equal to the asymptotic cell-formula with affine boundary conditions. \BBB In Lemma \ref{lemma:QnueqQ} and Lemma \ref{lemma:sureqaff} we use the following remark.

\begin{remark}\label{rem:densityofrationals} We point out that $\mathbb{S}^{d-1} \cap \mathbb{Q}^d$ is dense in $\mathbb{S}^{d-1}$. This follows from the fact that $\mathbb{Q}^{d-1}$ is dense in $\mathbb{R}^{d-1}$ and that the inverse of the stereographic projection $P_d \colon \mathbb{R}^{d-1} \to \mathbb{S}^{d-1} \setminus \{e_d\}$  is a rational and continuous function.
\end{remark}

\EEE

\begin{lemma}\label{lemma:QnueqQ} Let $\nu \in \mathbb{R}^{d}$. Then:
$\psi(\nu) =\varphi(\nu)$.
\end{lemma}
\begin{proof} Due the fact that both $\psi$ and \BBB $\varphi$ \EEE are positively homogeneous functions of degree one, it suffices to consider the case where $\nu \in \mathbb{S}^{d-1}$. Furthermore, since both functions are continuous, see Lemma \ref{lemma:propertiesofpsi}(ii) and Remark \ref{rem:continuity}, it suffices to prove the claim for $\nu \in \mathbb{S}^{d-1} \cap \mathbb{Q}^d$. For each such vector we can find $\nu_1,\ldots,\nu_{d-1} \in \mathbb{S}^{d-1} \cap \mathbb{Q}^d$ such that the set $\{\nu_1,\ldots,\nu_{d-1},\nu\}$ forms an orthonormal basis of $\mathbb{R}^d$. Then, there exists $\lambda \in \mathbb{N}$ such that
\begin{align}\label{eq:proplambda}
\lambda \nu_n = z_n \text{ for some } z_n \in \mathbb{Z}^d \text{  for all } n \in \{1,\ldots,d\}\,.
\end{align}
\BBB For $t \in (-1/2,1/2)$ we define an auxiliary function $\varphi_t \colon \mathbb{R}^d \to [0,+\infty)$ and for $\nu \in \mathbb{S}^{d-1}$  given by
\begin{align}\label{def:varphit}
\begin{split}
\varphi_t(\nu):= \lim_{\delta \to 0} \lim_{S\to +\infty} \frac{1}{S^{d-1}}\inf\big\{E(u,Q^\nu_S) \colon u \colon \mathcal{L}\to \{0,1\},  u(i)=\chi_{\{\langle\nu,i\rangle >tS\}} \text{ on } \mathcal{L} \setminus Q^\nu_{(1-\delta)S} \big\}\,.
\end{split}
\end{align}
Note that  $\varphi_0(\nu)=\varphi(\nu)$. \\
\noindent \begin{step}{1}($\varphi_t\geq \psi$ for all $t$) We show that for all $t\in (-1/2,1/2)$ we have 
\begin{align}\label{ineq:varphitpsi}
\varphi_t(\nu) \geq \psi(\nu)\,.
\end{align} \EEE
To this end let $\{\nu_1,\ldots,\nu_{d-1},\nu_d=\nu\} \subset \mathbb{S}^{d-1}\cap \mathbb{Q}^d$ be an orthonormal basis as previously described and let $1 \ll S_1 \ll S_2$. We assume that $S_1=\lambda T$, where $\lambda$ satisfies \eqref{eq:proplambda} and $T$ is given by {\rm (H1)}. Note that if $\lambda$ satisfies \eqref{eq:proplambda}, also $k\lambda$ satisfies \eqref{eq:proplambda} and therefore we can find a sequence $S_k=k \lambda T$ such that $S_k \to +\infty$ of the desired form. The existence of the limit in definition \eqref{def:varphi} of $\varphi$ permits us to assume that $S$ is of the specific form. \BBB Let $t\in (-1/2,1/2),\delta >0$ \EEE and $u_1 \colon \mathcal{L} \to \{0,1\}$ be such that \BBB $u_1(i)=\chi_{\{\langle\nu,i\rangle>tS\}}$ on $\mathcal{L} \setminus Q^\nu_{(1-\delta)S_1} $ \EEE and \BBB
\begin{align}\label{eq:u1minimum}
\begin{split}
E(u_1,Q^\nu_{S_1}) \leq  \inf\big\{E(u,Q^\nu_{S_1}) \colon u \colon \mathcal{L}\to \{0,1\},  u(i)=\chi_{\{\langle\nu,i\rangle>tS\}}\text{ on } \mathcal{L} \setminus Q^\nu_{(1-\delta)S_1}  \big\}+1\,.
\end{split}
\end{align} \EEE
Due to the assumption on $S_1$ and Lemma \ref{lemma:elementaryproperties}(vi), we have 
\begin{align}\label{eq:energyu1translate}
E(u_1(\cdot-z),Q^\nu_{S_1}(z))=E(u_1,Q^\nu_{S_1}) \text{ for all } z=\lambda T \sum_{n=1}^d k_n \nu_n, k \in \mathcal{L}\,.
\end{align}
Set (omitting the dependence on $S_1$ and $S_2$)
\begin{align*}
\BBB\mathcal{Z} := \left\{z =S_1\sum_{n=1}^d k_n \nu_n \colon k \in \mathbb{Z}^d, Q^\nu_{S_1}(z) \subset Q^\nu_{(1-\delta)S_2}  \right\}\,.\EEE
\end{align*}
We define $u_2 \colon \mathcal{L} \to \mathbb{R}$ by
\begin{align}\label{def:u2}
u_2(i) = \begin{cases} S_1\left(u_1(i-z)-\frac{1}{2}\right) + \langle \nu, z\rangle &\text{if } z \in \mathcal{Z}, i \in Q^\nu_{S_1}(z)\,,\\
\langle \nu,i\rangle &\text{otherwise.}
\end{cases}
\end{align}
We  claim that 
\begin{align}\label{ineq:u2diff}
|u_2(i)-u_2(j)| \leq C( S_1 + |i-j|)\,.
\end{align}
We postpone the proof of \eqref{ineq:u2diff} to the end of Step 1.
By the definition of $u_2$, it is clear that
\begin{align}\label{ineq:u2admissible}
\inf\left\{E(u,Q^\nu_{S_2}) \colon u \colon \mathcal{L} \to \mathbb{R}, u(i)=\langle\nu,i\rangle \text{ on } \mathcal{L} \setminus Q^\nu_{(1-\delta)S_2}\right\}\leq E(u_2,Q^\nu_{S_2})\,.
\end{align}
It remains to show that 
\begin{align}\label{ineq:u2u1}
E(u_2,Q^\nu_{S_2}) \leq \frac{S_2^d}{S_1^{d-1}}E(u_1,Q^\nu_{S_1}) + C\delta S_2^d+CS_1^2S_2^{d-1}+ \frac{S_2^d}{\delta}C_{S_1}^\delta\,,
\end{align}
where $C_{S_1}^\delta \to 0$ as $S_1 \to +\infty$.
In fact, once we have shown \eqref{ineq:u2u1}, Step 1 follows from \eqref{ineq:u2admissible} and \eqref{eq:u1minimum} by dividing by $S_2^d$ and letting first $S_2\to+\infty$, then $S_1\to+\infty$, and finally $\delta \to 0$. We are left to prove \eqref{ineq:u2u1}. \BBB In order to prove \eqref{ineq:u2u1} we introduce 
\begin{align}\label{def:r1r2}
\begin{split}
&r_1 = r_1( S_1,S_2,\delta)=(1-\delta)S_2- 3\sqrt{d}S_1\,, \\& r_2 = r_2(S_1,S_2,\delta)=(1-\delta)S_2+ 3\sqrt{d}S_1\,.
\end{split}
\end{align}
 We  use Lemma \ref{lemma:elementaryproperties}{\rm (iv)} to obtain 
\begin{align}\label{eq:u2u1splitestimate}
E(u_2,Q^\nu_{S_2})= E(u_2,Q^\nu_{r_1}) +  E(u_2,Q^\nu_{r_2}\setminus Q^\nu_{r_1}) + E(u_2,Q^\nu_{S_2} \setminus Q^\nu_{r_2})
\end{align}
and we estimate the \BBB three terms \EEE on the right hand side separately. We claim that
\begin{align}\label{ineq:firstrhs}
 E(u_2,Q^\nu_{r_1})  \leq \frac{S_2^d}{S_1^{d-1}}E(u_1,Q^\nu_{S_1}) + \frac{S_2^d}{\delta}C_{S_1}^\delta \,,
\end{align} 
where $C_{S_1}^\delta \to 0$ as $S_1\to +\infty$.
\EEE
\BBB Indeed, if $i \in Q^\nu_{S_1}(z)$  such that $z= S_1\sum_{n=1}^dk_n \nu_n\in \mathcal{Z}$ and $Q^\nu_{S_1}(z) \cap Q^\nu_{r_1} \neq \emptyset$ then \EEE
\begin{align}\label{eq:u2u1shifted}
u_2(i) = S_1\left(u_1(i-z)-\frac{1}{2}\right)+\langle \nu,z\rangle \text{ for all } i\in Q^\nu_{(1+\delta)S_1}(z)\,.
\end{align}
Due to \eqref{def:u2}, this is clearly true for $i \in Q^\nu_{S_1}(z)$, while if $i \in Q^\nu_{(1+\delta)S_1}(z) \setminus Q^\nu_{S_1}(z)$ there exists $z^\prime = S_1\sum_{n=1}^dk_n\nu_n \in \mathcal{Z}$, $k \in \mathbb{Z}^d$ such that  $||k-k^\prime||_\infty=1$, $Q^\nu_{S_1}(z^\prime) \subset Q^\nu_{r_1}$, and $i \in Q^\nu_{S_1}(z^\prime)\setminus Q^\nu_{(1-\delta)S_1}(z^\prime)$. Then, due to the boundary conditions of $u_1$, we have \BBB
\begin{align*}
u_2(i) &= S_1\left(u_1(i-z^\prime)-\frac{1}{2}\right) +\langle\nu,z^\prime\rangle =  S_1\left(\chi_{\{\langle\nu,i-z^\prime\rangle>tS\}}-\frac{1}{2}\right) +\langle\nu,z \rangle +\langle\nu,z^\prime-z \rangle \\& = S_1\left(\chi_{\{\langle\nu,i-z\rangle>tS\}}-\frac{1}{2}\right) +\langle\nu,z \rangle = S_1\left(u_1(i-z)-\frac{1}{2}\right) +\langle\nu,z \rangle \,.
\end{align*} \EEE
Here, the third equality follows, from the fact that $||k-k^\prime||_\infty=1$ and therefore $\langle\nu,z^\prime-z\rangle \in \{-S_1,0,S_1\} $. To obtain the previous equality, we distinguish the following two cases: \BBB
\begin{align*}
&\langle\nu,z^\prime-z\rangle = \pm S_1 \implies \chi_{\{\langle\nu,i-z^\prime\rangle>tS\}}-\chi_{\{\langle\nu,i-z\rangle>tS\}}=\mp 1\text{ and }  \\&\langle\nu,z^\prime-z\rangle = 0\implies \chi_{\{\langle\nu,i-z^\prime\rangle>tS\}} =\chi_{\{\langle\nu,i-z\rangle>tS\}}\,. 
\end{align*} \EEE
Now \eqref{eq:u2u1shifted} together with \eqref{eq:energyu1translate} implies for $z \in \mathcal{Z}$ such that $Q^\nu_{S_1}(z) \cap Q^\nu_{r_1} \neq \emptyset$
\begin{align}\label{ineq:Eu2split}
\nonumber
E(u_2,Q^\nu_{S_1}(z)) &= \sum_{i \in \mathcal{L} \cap Q^\nu_{S_1}(z)} \sum_{j \in \mathcal{L}\cap Q^\nu_{(1+\delta)S_1}(z)} c_{i,j}(u_2(i)-u_2(j))^+\\&\quad+ \sum_{i \in \mathcal{L}\cap Q^\nu_{S_1}(z)} \sum_{j \in \mathcal{L}\setminus Q^\nu_{(1+\delta)S_1}(z)} c_{i,j}(u_2(i)-u_2(j))^+  \\&\leq S_1E(u_1,Q^\nu_{S_1}(z))+ \sum_{i \in \mathcal{L}\cap Q^\nu_{S_1}(z)} \sum_{j \in \mathcal{L}\setminus Q^\nu_{(1+\delta)S_1}(z)} c_{i,j}|u_2(i)-u_2(j)| \nonumber\\&= S_1 E(u_1,Q^\nu_{S_1})+ \sum_{i \in \mathcal{L} \cap Q^\nu_{S_1}(z)} \sum_{j \in \mathcal{L}\setminus Q^\nu_{(1+\delta)S_1}(z)} c_{i,j}|u_2(i)-u_2(j)|\,.\nonumber
\end{align}
We estimate the second term on the right hand side of \eqref{ineq:Eu2split} to obtain \eqref{ineq:firstrhs}. Note for $i \in Q^\nu_{S_1}(z)$, $j \in \mathcal{L}\setminus Q^\nu_{(1+\delta)S_1}$ we have $|i-j|\geq \delta S_1/2$  and thus, due to \eqref{ineq:u2diff}, we obtain 
\begin{align}\label{ineq:differenceestimateu2}
|u_2(i)-u_2(j)|\leq \frac{C}{\delta}|i-j| \text{ for all } i \in Q^\nu_{S_1}(z), j \in \mathcal{L}\setminus Q^\nu_{(1+\delta)S_1}(z)
\end{align}
for some $C>0$ independent of $S_1$, $S_2$ and $\delta$. Now, we get
\begin{align*}
\sum_{i \in Q^\nu_{S_1}(z)} \sum_{j \in \mathcal{L}\setminus Q^\nu_{(1+\delta)S_1}(z)} c_{i,j}|u_2(i)-u_2(j)| &\leq \frac{C}{\delta} \sum_{i \in \mathcal{L}\cap Q^\nu_{S_1}(z)} \underset{|i-j|\geq \delta S_1/2}{\sum_{j \in\mathcal{L}}} c_{i,j}|i-j| \\&\leq \#(\mathcal{L}\cap Q^\nu_{S_1}(z)) \max_{i \in \mathcal{L}} \underset{|i-j|\geq \delta S_1/2}{\sum_{j \in\mathcal{L}}} c_{i,j}|i-j|  \leq \frac{C_{S_1}^\delta}{\delta} S_1^d\,,
\end{align*}
where $C_{S_1}^\delta \to 0$ as $S_1\to +\infty$.
Hence, noting that for $z,z^\prime \in \mathcal{Z}$ such that $z\neq z^\prime$, we have $Q^\nu_{S_1}(z) \cap Q^\nu_{S_1}(z^\prime)=\emptyset$ and therefore $\#\mathcal{Z} \leq S_2^d/S_1^d$, we get
\begin{align*}
E(u_2, Q^\nu_{r_1}) \leq \underset{Q^\nu_{S_1}(z) \cap Q^\nu_{r_1} \neq \emptyset}{\sum_{z \in \mathcal{Z}}}\BBB E(u_2,Q^\nu_{S_1}(z))\EEE &\leq \#\mathcal{Z}(S_1 E(u_1,Q^\nu_{S_1}) +\frac{C_{S_1}^\delta}{\delta} S_1^d)  \\&\leq \frac{S_2^d}{S_1^{d-1}}E(u_1,Q^\nu_{S_1}) + \frac{C_{S_1}^\delta}{\delta} S_2^d \,,
\end{align*}
where $C_{S_1}^\delta \to 0$ as $S_1\to +\infty$.
This is \eqref{ineq:firstrhs}. Next, we prove 
\BBB 
\begin{align}\label{ineq:secondrhs}
E(u_2,Q^\nu_{r_2} \setminus Q^\nu_{r_1} ) \leq C S_1^2 S_2^{d-1}\,.
\end{align}
We use \eqref{ineq:u2diff} to obtain
\begin{align*}
E(u_2,Q^\nu_{r_2} \setminus Q^\nu_{r_1} ) = \sum_{i \in Q^\nu_{r_2} \setminus Q^\nu_{r_1}}\sum_{j \in \mathcal{L}}c_{i,j}(u_2(i)-u_2(j))^+ \leq CS_1\#(\mathcal{L}\cap Q^\nu_{r_2} \setminus Q^\nu_{r_1}) \max_{i \in \mathcal{L}} \sum_{j \in \mathcal{L}} c_{i,j} |i-j|\,,
\end{align*}
where we used that, owing to {\rm (L1)}, we have $|i-j|\geq c$ if $i\neq j$.
Using Lemma \ref{lemma:elementaryproperties}(v), {\rm (H2)}, and \eqref{def:r1r2} we have that $\#(\mathcal{L}\cap Q^\nu_{r_2} \setminus Q^\nu_{r_1})\leq CS_1S_2^{d-1}$ we obtain \eqref{ineq:secondrhs}.
As for the third term on the right hand side we prove
\begin{align}\label{ineq:thirdrhs}
E(u_2,Q^\nu_{S_2} \setminus Q^\nu_{r_2}) \leq C\delta S_2^d + \frac{C_{S_1}^\delta}{\delta}S_2^d\,,
\end{align}
where $C_{S_1}^\delta \to 0$ as $S_1 \to +\infty$. To this end we
 split the summation over $j$ to obtain
\begin{align}\label{eq:summsplitE}
\begin{split}
E(u_2,Q^\nu_{S_2} \setminus Q^\nu_{r_2})  &= \sum_{i \in \mathcal{L} \cap Q^\nu_{S_2} \setminus Q^\nu_{r_2}} \underset{|i-j|\leq S_1\delta/2}{\sum_{j \in \mathcal{L} }} c_{i,j} (u_2(i)-u_2(j))^+\\&+\sum_{i \in \mathcal{L} \cap Q^\nu_{S_2} \setminus Q^\nu_{r_2}} \underset{|i-j|>S_1\delta/2}{\sum_{j \in \mathcal{L}}} c_{i,j} (u_2(i)-u_2(j))^+\,.
\end{split}
\end{align}
Let us note first that  $ Q^\nu_{S_2} \setminus Q^\nu_{r_2}\subset Q^{\nu}_{S_2} \setminus Q^\nu_{(1-2\delta)S_2}$  and therefore, due to \ref{lemma:elementaryproperties}(v), we have
\begin{align}\label{ineq:cardQsetQ}
 \#(\mathcal{L}\cap Q^\nu_{S_2} \setminus Q^\nu_{r_2}) \leq C\delta S_2^d\,.
\end{align}
Now, for the first term on the right hand side of \eqref{eq:summsplitE}, employing \eqref{def:u2}, we note that $u_2(i)=\langle\nu,i\rangle$ and $u_2(j)=\langle\nu,j\rangle$. Hence, 
\begin{align}\label{ineq:firstsummEu2}
\begin{split}
 \sum_{i \in \mathcal{L} \cap Q^\nu_{S_2} \setminus Q^\nu_{r_2}} \underset{|i-j|\leq S_1\delta/2}{\sum_{j \in \mathcal{L} }} c_{i,j} |u_2(i)-u_2(j)| \leq \#(\mathcal{L}\cap Q^\nu_{S_2} \setminus Q^\nu_{r_2}) \max_{i \in \mathcal{L}} \sum_{j \in \mathcal{L}} c_{i,j}|i-j| \leq C\delta S_2^d\,,
 \end{split}
\end{align}
where we used {\rm (H2)} and \eqref{ineq:cardQsetQ}. For the second term on the right hand side of \eqref{eq:summsplitE}, we use \eqref{ineq:u2diff} and \eqref{ineq:cardQsetQ}  to obtain
\begin{align}\label{ineq:secondsummEu2}
\begin{split}
\sum_{i \in \mathcal{L} \cap Q^\nu_{S_2} \setminus Q^\nu_{r_2}} \underset{|i-j|>S_1\delta/2}{\sum_{j \in \mathcal{L}}} c_{i,j} |u_2(i)-u_2(j)| &\leq \frac{C}{\delta} \sum_{i \in \mathcal{L} \cap Q^\nu_{S_2} \setminus Q^\nu_{r_2}} \underset{|i-j|>S_1\delta/2}{\sum_{j \in \mathcal{L}}} c_{i,j}|i-j| \\&\leq \frac{C}{\delta} \#(\mathcal{L}\cap Q^\nu_{S_2} \setminus Q^\nu_{r_2})\max_{i\in \mathcal{L}}\underset{|i-j|>S_1\delta/2}{\sum_{j \in \mathcal{L}}} c_{i,j}|i-j| \\& \leq CS_2^d C_{S_1}^\delta\,,
\end{split}
\end{align}
where $C_{S_1}^\delta \to 0$ as $S_1\to +\infty$. Inequality \eqref{ineq:secondrhs} follows from \eqref{eq:summsplitE}, \eqref{ineq:firstsummEu2}, and \eqref{ineq:secondsummEu2}. Now \eqref{eq:u2u1splitestimate}, \eqref{ineq:firstrhs}, \eqref{ineq:secondrhs}, \eqref{ineq:thirdrhs} give \eqref{ineq:u2u1}.
 To conclude Step 1, it remains to prove \eqref{ineq:u2diff}.  There are four cases to consider:
\begin{itemize}
\item[(a)]$i=i_0 +z$, $j=j_0+z$ $i_0 \in Q^\nu_{S_1},j_0 \in Q^\nu_{S_1}$, $z \in \mathcal{Z}$;
\item[(b)] $i=i_0 +z$, $j=j_0+z^\prime$ $i_0 \in Q^\nu_{S_1}(z),j_0 \in Q^\nu_{S_1}(z^\prime)$, $z,z^\prime \in \mathcal{Z}$;
\item[(c)] $i=i_0 +z$,  $i_0 \in Q^\nu_{S_1}(z)$, $z \in \mathcal{Z}$ $j_0 \notin Q^\nu_{S_1}(z^\prime)$ for any $z^\prime \in \mathcal{Z}$;
\item[(d)] $i \notin Q^\nu_{S_1}(z)$ for any $z \in \mathcal{Z}$ and $j \notin Q^\nu_{S_1}(z^\prime)$ for any $z^\prime \in \mathcal{Z}$.
\end{itemize}
\noindent\emph{Case }{\rm (a)}: This case follows since $\|u_1\|_{L^\infty(Q^\nu_{S_1})}\leq 1$. \\
\noindent \emph{Case }{\rm (b)}:  Note that in the case where $ i = i_0 +z$, $j=j_0 +z'$ for some $i_0,j_0 \in Q^\nu_{S_1}$ and for some $z,z^\prime \in \mathcal{Z}$, we have 
\begin{align*}
|u_2(i)-u_2(j)| \leq |\langle \nu,z-z^\prime\rangle| +CS_1 \leq  |\langle \nu,z+i_0-z^\prime-j_0\rangle| + |i_0-j_0| +CS_1 \leq |i-j| + CS_1
\end{align*}
and therefore \eqref{ineq:u2diff} holds true.\\
\noindent \emph{Case }{\rm (c)}:  Note that in the case where $ i = i_0 +z$, \BBB $i_0 \in Q^\nu_{S_1}$,\EEE $z \in \mathcal{Z}$ and $j \notin Q^\nu_{S_1}(z)$ for any $z \in \mathcal{Z}$, we have
\begin{align*}
|u_2(i)-u_2(j)|  \leq C S_1 +  |\langle\nu,z-j\rangle| \leq C S_1 +  |\langle\nu,i-j\rangle|  +|i_0| \leq C S_1 + |i-j|\,.
\end{align*}
 Also here \eqref{ineq:u2diff} holds true.\\ 
\noindent  \emph{Case }{\rm (d)}: In this case $u_2(i)=\langle\nu,i\rangle $ and $u_2(j)=\langle\nu,j\rangle$ and therefore \eqref{ineq:u2diff} holds true. This shows \eqref{ineq:u2diff} in general.
\end{step}\\
\noindent \begin{step}{2}($\varphi_t=\psi$ for almost all $t$)  Given $\delta>0$, $t\in (-1/2,1/2)$ and $S \gg 1$ and set
\begin{align}\label{def:varphitSdelta}
\varphi(t,S,\delta):= \frac{1}{S^{d-1}}\inf\big\{E(u,Q^\nu_S) \colon u \colon \mathcal{L}\to \{0,1\},  u(i)=\chi_{\{\langle\nu,i\rangle>tS\}}\text{ on } \mathcal{L} \setminus Q^\nu_{(1-\delta)S} \big\}\,.
\end{align}
Then, for $\varepsilon>0$ we find $u_1\colon \mathcal{L} \to \mathbb{R}$ be such that $u_1(i) = \langle\nu,i\rangle$ for $i \in \mathcal{L} \setminus Q^\nu_{(1-\delta)S}$ and
\begin{align}\label{ineq:u1min}
E(u_1,Q^\nu_S) \leq S^d(\psi(\nu)+\varepsilon)\,.
\end{align}
Due to Lemma \ref{lemma:coarea}, there holds
\begin{align}\label{ineq:u1u2step2}
\begin{split}
 E(u_1,Q_{S}^\nu) =
\int_{-\infty}^{+\infty} 
E(\chi_{\{u_1> t\}},Q_{S}^\nu)\,\mathrm{d}t
&\ge
\int_{-S(1-\delta)/2}^{S(1-\delta)/2} 
E(\chi_{\{u_1> t\}},Q_{S}^\nu)\,\mathrm{d}t \\&= S\int_{-(1-\delta)/2}^{(1-\delta)/2} 
E(\chi_{\{u_1> tS\}},Q_{S}^\nu)\,\mathrm{d}t \,.
\end{split}
\end{align}
Note that for all $t \in (-(1-\delta)/2,(1-\delta)/2)$, due to $u_1(i)=\langle\nu,i\rangle$ we have that $\chi_{\{u_1> tS\}}(i) = \chi_{\{\langle\nu,i\rangle>tS\}}$ for $i \in \mathcal{L} \setminus Q^\nu_S $ and thus
\begin{align}\label{ineq:admissible}
\begin{split}
E(\chi_{\{u_1>tS\}},Q_{S}^\nu) \geq S^{d-1} \varphi(t,S,\delta)\,.
\end{split}
\end{align}
Therefore, 
\[
\varepsilon+\psi(\nu)\ge
\int_{-(1-\delta)/2}^{(1-\delta)/2} \varphi(t,S,\delta)\,\mathrm{d}t 
\]
and then thanks to Fatou's lemma we deduce for $\delta_0>0$
\begin{align*}
\varepsilon+\psi(\nu)
\ge\int_{-(1-\delta_0)/2}^{(1-\delta_0)/2} \varphi_t(\nu)\,\mathrm{d}t\,.
\end{align*}
After letting $\delta_0\to 0$  and $\varepsilon\to 0$, using Step 1, we obtain
\begin{align*}
\int_{-1/2}^{1/2} \varphi_t(\nu)\,\mathrm{d}t \leq \psi(\nu) \leq \varphi_t(\nu) \text{ for all } t\in (-1/2,1/2)\,.
\end{align*}
Hence, $\varphi_t(\nu) = \psi(\nu)$ for almost all $t\in (-1/2,1/2)$.
This concludes Step 2.
\end{step}\\
\begin{step}{3}($t \mapsto \varphi_t(\nu)$ is constant) To this end, let $t_1,t_2 \in (-1/2,1/2)$, $t_2<t_1$, let $\varepsilon>0,\delta>0$, $S\gg 1$, and let $u_1 \colon \mathcal{L} \to \{0,1\}$ be such that $u_1(i)=\chi_{\{\langle\nu,i\rangle >t_1S \}}$ on $\mathcal{L}\setminus Q^\nu_S$ and
\begin{align}\label{ineq:u1mincont}
E(u_1,Q^\nu_S) \leq S^{d-1}\left(\varphi_{t_1}(\nu) +\varepsilon\right)\,.
\end{align}
 We set $u_\nu^s(i) = \chi_{\{\langle\nu,i\rangle > s\}}$ and define $u_2 \colon \mathcal{L}\to \{0,1\}$ by 
\begin{align}\label{def:u2step3}
u_2(i) = u_1(i) + (u_\nu^{t_2S}(i)-u_\nu^{t_1S}(i))\chi_{(Q^\nu_{(1-\delta)S})^c}(i)\,.
\end{align}
It is obvious that $u_2(i) = u_\nu^{t_2S}(i)$ on $\mathcal{L}\setminus Q^\nu_{(1-\delta)S}$ and therefore 
\begin{align}\label{ineq:u2varphideltaSt}
E(u_2,Q^\nu_S)\geq S^{d-1}\varphi(t_2,S,\delta)\,.
\end{align}
Next, we show that
\begin{align}\label{ineq:u2u1step3}
E(u_2,Q^\nu_S)\leq E(u_1,Q^\nu_S) + C(\delta+|t_1-t_2|) S^{d-1} + \frac{C_S^\delta}{\delta}S^{d-1}\,,
\end{align}
where $C_S^\delta \to 0 $ as $S \to +\infty$. Now the claim follows by \eqref{ineq:u1mincont}, \eqref{ineq:u2varphideltaSt}, and \eqref{ineq:u2u1step3} by dividing with $S^{d-1}$ and letting first $S\to+\infty$, then
$\delta\to 0$, and eventually $\varepsilon\to 0$. It remains to prove \eqref{ineq:u2u1step3}. Here, we exploit Lemma \ref{lemma:elementaryproperties}(ii) and \eqref{def:u2step3} to deduce
\begin{align}
E(u_2,Q^\nu_S) \leq E(u_1,Q^\nu_S)  +E((u_\nu^{t_2S}-u_\nu^{t_1S})\chi_{(Q^\nu_{(1-\delta)S})^c},Q^\nu_S)\,.
\end{align}
We call $v=(u_\nu^{t_2S}-u_\nu^{t_1S})\chi_{(Q^\nu_{(1-\delta)S})^c}$ and note that it suffices to prove
\begin{align}\label{ineq:vStep3}
E(v,Q^\nu_S) \leq C(\delta +|t_1-t_2|) S^{d-1} + \frac{C_S^\delta}{\delta}S^{d-1}\,.
\end{align}
 We observe that
\begin{align}\label{eq:valuesv}
\{v=-1\} = \mathcal{L} \cap B_{t_1,t_2}^S\,, \text{ where } B_{t_1,t_2}^S:= \{x \in \mathbb{R}^d\setminus Q^\nu_{(1-\delta)S} \colon t_1S\leq \langle\nu,x\rangle< t_2S\}
\end{align} 
and  $\{v=0\}=\mathcal{L}\setminus B_{t_1,t_2}^S$.
Therefore
\begin{align}\label{eq:splitvStep3}
E(v,Q^\nu_S) = \sum_{i \in \{v=0\}\cap Q^\nu_S} \sum_{j \in \{v=-1\}} c_{i,j} \leq  \sum_{i \in \mathcal{L} \cap Q^\nu_S} \underset{|i-j|>\delta S}{\sum_{j \in \mathcal{L}}} c_{i,j}+ \sum_{i \in \{v=0\}\cap Q^\nu_S} \underset{|i-j|\leq \delta S}{\sum_{j \in \{v=-1\}}} c_{i,j}  \,.
\end{align}
As for the first term on the right hand side of \eqref{eq:splitvStep3}, we point out that, due to Lemma \ref{lemma:elementaryproperties}(v), we have
\begin{align}\label{ineq:firstv}
\sum_{i \in \mathcal{L} \cap Q^\nu_S} \underset{|i-j|>\delta S}{\sum_{j \in \mathcal{L}}} c_{i,j} \leq \frac{1}{\delta S} \#(\mathcal{L}\cap Q^\nu_S) \max_{i \in \mathcal{L}}\underset{|i-j|>\delta S}{ \sum_{j \in \mathcal{L}}} c_{i,j} |i-j| \leq \frac{C_S^\delta}{\delta} S^{d-1}\,,
\end{align}
where $C_S^\delta \to 0$ as $S\to +\infty$. Now, let us consider $j \in \mathcal{L}$ such that $|i-j| <\delta S$. For $i_0,j_0 \in Q_T$ and $\xi \in \mathbb{Z}^d$ set
 \begin{align}\label{def:Ai0j0xi}
 \begin{split}
 A_{i_0,j_0}^\xi :=\big\{(z,z^\prime) \in T\mathbb{Z}^d\times T\mathbb{Z}^d \colon z^\prime-z =\xi,\, &i=i_0+z \in \{v=0\} \cap Q^\nu_S,\\& j=j_0+z^\prime \in \{v=-1\}\big\}\,.
 \end{split}
 \end{align} 
 We observe that $A_{i_0,j_0}^\xi \subset \{(z,z+\xi) \colon z \in T\mathbb{Z}^d, \mathrm{dist}(z, \partial B_{t_1,t_2}^S \cap Q^\nu_{S}) \leq \sqrt{d}T+ |\xi| \} $ and thus for $|\xi|\leq C\delta S$ \begin{align*}
 \#  A_{i_0,j_0}^\xi\leq C(\delta + |t_1-t_2|)(|\xi|+T) S^{d-1}\,.
 \end{align*}
 Therefore, there holds 
 \begin{align}\label{eq:sumsplitS21}
\nonumber \sum_{i \in \{v=0\}\cap Q^\nu_S} \underset{|i-j|\leq \delta S}{\sum_{j \in \{v=-1\}}} c_{i,j} &\leq  \sum_{i_0,j_0 \in Q_T} \underset{|\xi|\leq \delta S}{\sum_{\xi \in T\mathbb{Z}^d}}\sum_{(z,z^\prime) \in A_{i_0,j_0}^\xi} c_{i_0+z,j_0+z^\prime}   \\&\leq C(\delta +|t_1-t_2|) S^{d-1}\sum_{i_0,j_0\in Q_T} \underset{|\xi|\leq \delta S_1}{\sum_{\xi \in T\mathbb{Z}^d}} c_{i_0,j_0+\xi} (|\xi|+T) \\&\leq C(\delta +|t_1-t_2|) S^{d-1}\max_{i\in \mathcal{L}\cap Q_T}\sum_{j \in \mathcal{L}} c_{i,j}|i-j| \leq C(\delta +|t_1-t_2|) S^{d-1}\,.\nonumber
\end{align}
Here we used that $|\xi| \leq |i-j| + |i_0-j_0|\leq |i-j| +\sqrt{d}T$ and the fact that $|i-j| \geq c = T \cdot (c/T)$ for all $i \neq j$, where $c/T>0$ is a fixed constant.
This together with \eqref{eq:splitvStep3} and \eqref{ineq:firstv} implies \eqref{ineq:vStep3} and therefore $\varphi_{t_2}(\nu) \leq \varphi_{t_1}(\nu) +C|t_1-t_2|$ for $t_2 <t_1$. Due to Step 2, for any $t \in (-1/2,1/2)$ we can find $t_n \to t$ such that $\varphi_{t_n}(\nu)=\psi(\nu)$ and $\varphi_{t}(\nu) \leq \varphi_{t_n}(\nu) +|t-t_n|= \psi(\nu)+|t-t_n|$. Letting $t_n\to t$ we obtain $\varphi_t(\nu)\leq \psi(\nu)$. This together with Step 1 shows $\varphi_t(\nu)=\psi(\nu)$.
\end{step}\\
\noindent Due to Step 3 we have that $\varphi(\nu)=\varphi_0(\nu)=\psi(\nu)$. This concludes the proof.
\end{proof}
\EEE

In the next Lemma we show that, assuming affine boundary conditions, the calculation of the asymptotic cell formula with respect to the coordinate cube and the calculation of the asymptotic cell formula with respect to the rotated cube are equivalent. 

\begin{lemma}\label{lemma:sureqaff} Let $\nu \in\mathbb{R}^d$. Then:
  $\psi(\nu) = \phi(\nu)$.
\end{lemma}
\begin{proof} \BBB After reducing to rational directions, we define a sequence of cell problems defined on cubes $Q^\nu_{\lambda T}$ and exploit Lemma \ref{lemma: per=aff} to show that it suffices to compare the two sequences of cell problems with periodic boundary conditions. \EEE

\noindent \BBB Since both $\psi$ and $\phi$ \EEE are positively homogeneous functions of degree one (cf.~\eqref{def:psi} and Lemma \ref{lemma:propertiesofphi}(ii)) it suffices to consider the case where $\nu \in \mathbb{S}^{d-1}$. Thanks to Lemma \ref{lemma:propertiesofpsi}(ii) and  Lemma \ref{lemma:propertiesofphi}(iii) both functions are continuous. Thus it suffices to prove the claim for  $\nu \in \mathbb{S}^{d-1} \cap \mathbb{Q}^d$. For each such vector we can find $\{\nu_1,\ldots,\nu_{d-1}\} \in \mathbb{S}^{d-1} \cap \mathbb{Q}^d$ such that the set $\{\nu_1,\ldots,\nu_{d-1},\nu_d=\nu\}$ forms an orthonormal basis of $\mathbb{R}^d$. For such an orthonormal basis, it is clear that there exists $\lambda \in \mathbb{N}$ such that
\begin{align}\label{eq:proplambdapsieqphi}
\lambda \nu_n =z_n \text{ for some } z_n \in \mathbb{Z}^d \text{ for all } n \in \{1,\ldots,d\}\,.
\end{align}\\

\BBB\noindent Let $\{\nu_1,\ldots,\nu_{d-1},\nu_d=\nu\}\subset \mathbb{S}^{d-1} \cap \mathbb{Q}^d$ be the orthonormal basis described previously. For $\nu \in \mathbb{S}^{d-1}\cap \mathbb{Q}^d$ fixed, we set
\begin{align}\label{def:psiper}
\psi_{\mathrm{per}}(\nu) :=\lim_{k \to \infty} \frac{1}{(\lambda kT)^d}\inf\{&E(u,Q^\nu_{\lambda kT}) \colon u \colon \mathcal{L} \to \mathbb{R},\\& u(\cdot +\lambda k T \nu_n)-\langle\nu,\cdot+\lambda k T \nu_n \rangle =  u(\cdot)-\langle\nu,\cdot\rangle \text{ for all } n=1,\ldots,d\}\nonumber
\end{align}
\noindent \begin{step}{1}($\psi_{\mathrm{per}}=\psi$) We claim that 
\begin{align}
\psi_{\mathrm{per}}(\nu)=\psi(\nu)\,.
\end{align}
This follows exactly as the proof of Lemma \ref{lemma: per=aff} by replacing the coordinate cubes with the cubes $Q^\nu_{\lambda kT}$.
\end{step}\\
\noindent \begin{step}{2}($\psi_{\mathrm{per}}=\phi_{\mathrm{per}}$) Due to Step 1 and Lemma \ref{lemma: per=aff} it suffices to show that
\begin{align}
\psi_{\mathrm{per}}(\nu)=\phi_{\mathrm{per}}(\nu)\,.
\end{align}
\noindent \begin{step}{2.1}(Proof of '$\geq$') To this end, let $\varepsilon>0$, $k \in \mathbb{N}$ be big enough, and $u_k \colon \mathcal{L} \to \mathbb{R}$ be such that $u_k(\cdot +\lambda k T \nu_n)-\langle\nu,\cdot+\lambda k T \nu_n \rangle =  u_k(\cdot)-\langle\nu,\cdot\rangle \text{ for all } n=1,\ldots,d$ and 
\begin{align}\label{ineq:almostminpsiper}
E(u_k,Q_{\lambda kT}^\nu) \leq (\lambda kT)^d( \psi_{\mathrm{per}}(\nu) +\varepsilon)\,.
\end{align}
Thanks to \eqref{eq:proplambdapsieqphi} and the fact that $\{\nu_1,\ldots,\nu_d\}$ is a basis, we have that $u_k \in \mathcal{A}_{\mathrm{per}}(Q_{kmT};\mathbb{R})$ for some $m \in \mathbb{N}$ depending only on $\{\nu_1,\ldots,\nu_d\}$ and $\lambda$. Fix now $M\in \mathbb{N}$ such that $M\gg m$ and let
\begin{align*}
\mathcal{Z} := \left\{z=\lambda kT \sum_{n=1}^d \mu_n \nu_n \colon \mu \in \mathbb{Z}^d,  Q^\nu_{\lambda kT}(z) \cap Q_{MkT}\neq \emptyset\right\}\,.
\end{align*}
Since,  $u_k \in \mathcal{A}_{\mathrm{per}}(Q_{kmT};\mathbb{R})$ we have that  $u_k \in \mathcal{A}_{\mathrm{per}}(Q_{kMT};\mathbb{R})$ and thus 
\begin{align}\label{ineq:ukadmi}
\inf\left\{E(u,Q_{kmT})\colon u \colon \mathcal{L}\to \mathbb{R}, u(\cdot)-\langle\nu,\cdot\rangle \in \mathcal{A}_{\mathrm{per}}(Q_{kMT};\mathbb{R})\right\}\leq  E(u_k,Q_{kMT}) \,.
\end{align}
Due to the periodicity of $u_k$, the assumption on $\lambda$, and Lemma \ref{lemma:elementaryproperties}(vi), we have that
\begin{align}\label{eq:translationinvariance}
E(u_k,Q^\nu_{\lambda kT}(z)) =E(u_k(\cdot-z),Q^\nu_{\lambda kT}(z)) = E(u_k,Q^\nu_{\lambda kT}) \text{ for all } z = \lambda kT \sum_{n=1}^d \mu_n \nu_n, \mu \in \mathbb{Z}^d\,.
\end{align}
Note that for $z \in \mathcal{Z}$ we have $Q_{\lambda kT}^\nu(z) \subset Q_{(M+\sqrt{d}\lambda)kT}$ and thus  $\#\mathcal{Z} \leq (M^d + CM^{d-1}\lambda kT)/\lambda^d$. Therefore,
\begin{align}
E(u_k,Q_{kMT}) \leq \sum_{z \in \mathcal{Z}} E(u_k,Q^\nu_{\lambda kT}(z)) \leq \#\mathcal{Z} E(u_k,Q^\nu_{\lambda kT})\leq  \frac{M^d + CM^{d-1}\lambda kT}{\lambda^d} E(u_k,Q^\nu_{\lambda kT})\,.
\end{align}
Dividing by $(kMT)^d$, letting first $M$ tend to $+\infty$, then $k$ to $+\infty$, and lastly $\varepsilon \to 0$, and noting \eqref{ineq:almostminpsiper} as well as \eqref{ineq:ukadmi}, we obtain the conclusion of Step 2.1.
\end{step}\\
\noindent \begin{step}{2.2}(Proof of '$\leq$') The reverse inequality follows as in Step 2.1 by noting that if $u \in \mathcal{A}_{\mathrm{per}}(Q_{kT};\mathbb{R})$ for some $k \in \mathbb{N}$, then, due to \eqref{eq:proplambdapsieqphi}, we have that $u(\cdot-\lambda m T \nu_n)-\langle\nu,\cdot - \lambda m T \nu_n\rangle=u(\cdot)-\langle\nu,\cdot \rangle$ for some $m \in \mathbb{N}$ and all $n=1,\ldots,d$. This allows us to perform the same construction as in Step 2.1.
\end{step}
This concludes the proof.
\end{step}
\EEE
\end{proof}

\begin{proof}[Proof of Proposition \ref{prop:main}] 
Our goal is to prove
\begin{align}\label{eq:varphiT}
\varphi(\nu)= \frac{1}{T^d}\inf\left\{E(u,Q_T) \colon u \colon \mathcal{L}\to \mathbb{R},  u(\cdot)-\langle\nu,\cdot\rangle \BBB \in \mathcal{A}_{\mathrm{per}}(Q_T;\mathbb{R})\EEE\right\}\,
\end{align}
for all $\nu \in \mathbb{R}^d$. Due to Lemma \ref{lemma:sureqaff}, Lemma \ref{lemma:QnueqQ}, and Lemma \ref{lemma: per=aff}, we have
\begin{align}\label{eq:varphi=phiper}
\varphi(\nu) = \psi(\nu)=\phi(\nu) = \phi_{\rm per}(\nu)\,.
\end{align}
Additionally, Lemma \ref{lemma:Treduction} ensures that
\begin{align*}
\phi_{\rm per}(\nu) = \frac{1}{T^d}\inf\left\{E(u,Q_T) \colon u \colon \mathcal{L}\to \mathbb{R},  u(\cdot)-\langle\nu,\cdot\rangle\BBB \in \mathcal{A}_{\mathrm{per}}(Q_T;\mathbb{R})\right\}\,.
\end{align*}
This shows \eqref{eq:varphiT} and concludes the proof.
\end{proof}

\section{Crystallinity of the homogenized surface energy density}\label{sec:Crystallinity} This section is devoted to the proof of Theorem \ref{theorem:main}. We assume throughout this section that assumptions \BBB {\rm (L1)}, {\rm (L2)} \EEE and {\rm (H1)}, {\rm (H3)} are satisfied.

We define the set of edges $\mathcal{E}$ by 
\begin{align}\label{def:edges}
\mathcal{E}=\{(i,j) \in (\mathcal{L}\cap Q_T) \times \mathcal{L} : c_{i,j}\neq 0\} \text{ and } N = \# \mathcal{E}\,.
\end{align}

\begin{proof}[Proof of Theorem \ref{theorem:main}]  We divide the proof into three steps. First, we derive a dual representation of $\varphi$. Then, using this representation, we show that $\varphi$ is crystalline. \\
\noindent \begin{step}{1}(Dual representation) 
 We define
\begin{align}\label{def:C}
\begin{split}
\mathcal{C}=\Big\{\BBB \alpha_{i,j}\in [0,c_{i,j}]  \EEE\colon &\alpha_{i+Tz,j+Tz}=\alpha_{i,j}\text{ for all } z \in \mathcal{L},\\&\sum_{j \in \mathcal{L}}(\alpha_{j,i}-\alpha_{i,j})=0 \text{ for all } i \in Q_T\cap\mathcal{L} \Big\}\,.
\end{split}
\end{align}
Our goal is to prove
\begin{align}\label{eq:dual representation}
\varphi(\nu) =\frac{1}{T^d} \sup_{(\alpha_{i,j})_{i,j} \in \mathcal{C}} \left\langle \nu, \sum_{i \in \mathcal{L} \cap Q_T }\sum_{j \in \mathcal{L}}\alpha_{i,j}(i-j)\right\rangle\,.
\end{align}
Let $\nu \in \mathbb{R}^d$.
Due to Proposition, \ref{prop:main} there holds
\begin{align*}
\varphi(\nu) = \phi_{\mathrm{per}}(\nu) &= \frac{1}{T^d}\inf\left\{E(u,Q_T) \colon  u(\cdot)-\langle\nu,\cdot\rangle \BBB \in \mathcal{A}_\mathrm{per}(Q_T;\mathbb{R})\EEE\right\} \\&= \frac{1}{T^d}\inf\left\{E(u+\langle\nu,\cdot\rangle,Q_T) \colon u\BBB \in \mathcal{A}_\mathrm{per}(Q_T;\mathbb{R})\right\}
\\&=\BBB \frac{1}{T^d}\inf_{u \in \mathcal{A}_{\mathrm{per}}(Q_T;\mathbb{R})} \sum_{i \in Q_T} \sum_{j \in \mathcal{L}} c_{i,j}(u(i)-u(j)+\langle\nu,i-j\rangle)^+\,. \EEE
\end{align*}
Note that, we can write
\begin{align}\label{eq:minmaxGammanu}
\phi_{\mathrm{per}}(\nu) = \frac{1}{T^d} \BBB\inf_{u \in \mathcal{A}_{\mathrm{per}}(Q_T;\mathbb{R})} \EEE\, \underset{\alpha_{i+Tz,j+Tz}=\alpha_{i,j}}{\sup_{0 \leq \alpha_{i,j} \leq c_{i,j}}}\sum_{i\in \mathcal{L}\cap Q_T} \sum_{j \in \mathcal{L}} \alpha_{i,j}(u(i)-u(j)+\langle\nu,i-j\rangle)\,.
\end{align}
\BBB To see this, we observe that for all $u  \in \mathcal{A}_\mathrm{per}(Q_T;\mathbb{R})$ and all $0\leq \alpha_{i,j}\leq c_{i,j}$ such that $\alpha_{i+Tz,j+Tz}=\alpha_{i,j}$ for all $z \in \mathbb{Z}^d$ we have
\begin{align*}
\alpha_{i,j}(u(i)-u(j)+\langle\nu,i-j\rangle)  \leq c_{i,j} (u(i)-u(j)+\langle\nu,i-j\rangle)^+
\end{align*}
with equality for
\begin{align*}
\alpha_{i,j} = \begin{cases} c_{i,j} &\text{if } (u(i)-u(j)+\langle\nu,i-j\rangle) \geq 0,\\
0 &\text{otherwise.}
\end{cases}
\end{align*}\EEE
Given $0 \leq \alpha_{i,j}\leq c_{i,j}$ such that $\alpha_{i+Tz,j+Tz}=\alpha_{i,j}$ for all $z\in \mathbb{Z}^d$, and $u \colon \mathcal{L}\to \mathbb{R}$ $T$-periodic, we have
\begin{align*}
\sum_{i\in\mathcal{L}\cap Q_T} \sum_{j \in \mathcal{L}}\alpha_{i,j}(u(i)-u(j)) &= \sum_{i\in \mathcal{L}\cap Q_T} \sum_{j \in \mathcal{L}} \alpha_{i,j}u(i)- \sum_{i\in \mathcal{L}\cap Q_T} \sum_{j \in \mathcal{L}} \alpha_{i,j}u(j) \\&= \sum_{i\in \mathcal{L}\cap Q_T} \sum_{j \in \mathcal{L}} \alpha_{i,j}u(i) - \sum_{j \in \mathcal{L}\cap Q_T}\sum_{z \in \mathbb{Z}^d} \sum_{i\in \mathcal{L}\cap Q_T} \alpha_{i,j+Tz}u(j+Tz) \\&= \sum_{i\in \mathcal{L}\cap Q_T} \sum_{j \in \mathcal{L}} \alpha_{i,j}u(i) - \sum_{j \in \mathcal{L}\cap Q_T}\sum_{z \in \mathbb{Z}^d} \sum_{i\in \mathcal{L}\cap Q_T} \alpha_{i-Tz,j}u(j)\\&=
 \sum_{i\in \mathcal{L}\cap Q_T} \sum_{j \in \mathcal{L}} \alpha_{i,j}u(i) - \sum_{j \in \mathcal{L}\cap Q_T}\sum_{i\in \mathcal{L}} \alpha_{i,j}u(j)
 \\&=\sum_{i\in \mathcal{L}\cap Q_T} \sum_{j \in \mathcal{L}} (\alpha_{i,j}-\alpha_{j,i})u(i)\,.
\end{align*}
Note that, since in all steps the sum over $i$ and, due to {\rm (H3)}, the sum over $j$ runs over a finite index set, the order of summation can be changed without changing the value of the various sums.
This implies that, given $0 \leq \alpha_{i,j}\leq c_{i,j}$ such that $\alpha_{i+Tz,j+Tz}=\alpha_{i,j}$ for all $z\in \mathcal{L}$, we have
\begin{align}\label{eq:alphaix}
\BBB \inf_{u \in \mathcal{A}_{\mathrm{per}}(Q_T;\mathbb{R})} \sum_{i\in\mathcal{L}\cap Q_T} \sum_{j \in \mathcal{L}}\alpha_{i,j}\EEE (u(i)-u(j)) = \begin{cases} 0 &\text{if }\displaystyle \sum_{j \in \mathcal{L}}(\alpha_{i,j}-\alpha_{j,i})=0 \text{ for all } i \in Q_T\cap\mathcal{L}\,,\\
-\infty &\text{otherwise.}
\end{cases}
\end{align}
Hence, using  \eqref{eq:varphi=phiper}, \eqref{def:C}, \eqref{eq:minmaxGammanu}, and \eqref{eq:alphaix}, we obtain \BBB
\begin{align}\label{dualineq}
\varphi(\nu) \geq \frac{1}{T^d} \sup_{(\alpha_{i,j})_{i,j} \in \mathcal{C}} \left\langle \nu, \sum_{i \in \mathcal{L} \cap Q_T }\sum_{j \in \mathcal{L}}\alpha_{i,j}(i-j)\right\rangle\,.
\end{align} 
As for the other inequality in finite dimension, that is when $c_{i,j}^R$ is such that $c_{i,j}^R=c_{i,j}$ if $|i-j|<R$ and $c_{i,j}^R=0$ if $|i-j|\geq R$, the equality is true due to \cite[Corollary 31.2.1]{Rockafellar}. More precisely, 
we obtain:\begin{multline}\label{eq:R}
              \varphi_R(\nu) =  \frac{1}{T^d}\inf_{u \in \mathcal{A}_{\mathrm{per}}(Q_T;\mathbb{R})} \sum_{i \in Q_T} \sum_{j \in \mathcal{L}} c_{i,j}^R(u(i)-u(j)+\langle\nu,i-j\rangle)^+
              \\= \frac{1}{T^d} \sup_{(\alpha_{i,j})_{i,j} \in \mathcal{C}_R} \left\langle \nu, \sum_{i \in \mathcal{L} \cap Q_T }\sum_{j \in \mathcal{L}}\alpha_{i,j}(i-j)\right\rangle\,,
\end{multline}
where $\mathcal{C}_R=\mathcal{C}\cap \prod_{(i,j)}[0,c_{i,j}^R]$.
Note that $\mathcal{C}_R\subset \mathcal{C}$ and thus by~\eqref{dualineq} and~\eqref{eq:R}, we have
\begin{align*}
  \varphi_R(\nu) \leq  \frac{1}{T^d} \sup_{(\alpha_{i,j})_{i,j} \in \mathcal{C}} \left\langle \nu, \sum_{i \in \mathcal{L} \cap Q_T }\sum_{j \in \mathcal{L}}\alpha_{i,j}(i-j)\right\rangle
\leq   \varphi(\nu) \,.
\end{align*}
In addition, for all $u \in \mathcal{A}_{\mathrm{per}}(Q_T;\mathbb{R})$ we have
\begin{align*}
 \lim_{R\to +\infty}\sum_{i \in Q_T} \sum_{j \in \mathcal{L}} c_{i,j}^R(u(i)-u(j)+\langle\nu,i-j\rangle)^+= \sum_{i \in Q_T} \sum_{j \in \mathcal{L}} c_{i,j}(u(i)-u(j)+\langle\nu,i-j\rangle)^+\,,
\end{align*}
monotonically in $R$. Hence, due to $\Gamma$-convergence of monotone sequences we have
\begin{align*}
\varphi(\nu)=\lim_{R\to +\infty} \varphi_R(\nu) \leq \frac{1}{T^d} \sup_{(\alpha_{i,j})_{i,j} \in \mathcal{C}} \left\langle \nu, \sum_{i \in \mathcal{L} \cap Q_T }\sum_{j \in \mathcal{L}}\alpha_{i,j}(i-j)\right\rangle\,.
\end{align*}
This shows \eqref{eq:dual representation}.
\EEE
\end{step}\\ 
\noindent \begin{step}{2}(Crystallinity) By Remark \ref{rem:duality}, we have
\begin{align*}
\varphi(\nu) = \sup_{\zeta \in W_{\varphi}} \langle\nu,\zeta\rangle\,.
\end{align*}
So that, by \eqref{eq:dual representation} 
\begin{align}\label{eq:Wvarphi}
W_\varphi = \left\{\frac{1}{T^d}\sum_{i \in \mathcal{L}\cap Q_T} \sum_{j \in \mathcal{L}}\alpha_{i,j}(i-j)\colon (\alpha_{i,j})_{i,j} \in \mathcal{C}  \right\}\,,
\end{align}
with $\mathcal{C}$ given in \eqref{def:C}. 
 Recall $N$ and $\mathcal{E}$ defined in \eqref{def:edges}. Define $L \colon \mathbb{R}^N \to \mathbb{R}^d$ by
\begin{align}\label{def:L}
L\left(\alpha_{i,j}\right)_{(i,j) \in \mathcal{E}} =\frac{1}{T^d} \sum_{(i,j) \in \mathcal{E}} \alpha_{i,j}(i-j)\,.
\end{align}
\BBB Hence, we observe that 
\begin{align*}
W:=\left\{\frac{1}{T^d}\sum_{i \in \mathcal{L}\cap Q_T} \sum_{j \in \mathcal{L}}\alpha_{i,j}(i-j)\colon (\alpha_{i,j})_{i,j} \in \mathcal{C}  \right\} = L\left(\mathcal{C}\right)\,,
\end{align*}
where $\mathcal{C}$ is a convex closed set as the intersection of two convex and closed sets. Thus, as the image of the convex and closed set $\mathcal{C}$ through the linear map $L$, $W$ is a closed and convex set. Then \eqref{eq:Wvarphi} follows by fenchel duality, since
\begin{align*}
\varphi(\nu) = \mathrm{I}_{W}^*(\nu) = \mathrm{I}_{W_\varphi}^*(\nu)\,,
\end{align*}
where $\mathrm{I}_A \colon \mathbb{R}^d\to [0,+\infty]$ denotes the support function denotes the support function of the set $A$ given by
\begin{align*}
\mathrm{I}_A(\zeta)=\begin{cases} 0 &\text{if } \zeta \in A\,,\\
+\infty &\text{otherwise.}
\end{cases}
\end{align*}
Now, since both $W_{\varphi}$ and $W$ are closed and convex, we have
\begin{align*}
\mathrm{I}_{W_{\varphi}}(\zeta) = \mathrm{I}_{W_{\varphi}}^{**}(\zeta) = \varphi^*(\zeta)= \mathrm{I}_{W}^{**} = \mathrm{I}_{W}\,.
\end{align*}
This shows  \eqref{eq:Wvarphi}.
 \EEE
Furthermore, we find
\begin{align}\label{eq:WphiL}
W_{\varphi} = L \left( V\cap\prod_{(i,j) \in \mathcal{E}} [0,c_{i,j}]  \right)\,,
\end{align}
where $V \subset \mathbb{R}^N$ is a linear subspace of co-dimension $k:=\#(Q_T \cap \mathcal{L})-1$ given by
\begin{align}\label{def:V}
V=\left\{\alpha_{i,j} \in \mathbb{R}^N\colon \sum_{j \in \mathcal{L}}(\alpha_{i,j}-\alpha_{j,i})=0 \text{ for all } i \in Q_T\cap\mathcal{L}\right\}\,.
\end{align}
Hence, due to \eqref{eq:WphiL}, $W_{\varphi}$ is the image of the linear map $L$, given in \eqref{def:L}, of a $N$-dimensional polytope $\prod_{(i,j) \in \mathcal{E}}[0,c_{i,j}]$ intersected with the linear subspace $V$, given in \eqref{def:V}. The intersection of a cube with a linear subspace is a polytope, and thus also its image through a linear map. This proves that $\varphi$ is crystalline. \end{step}\\
\noindent \begin{step}{3}(Estimate on the number of vertices) Our goal is to prove that
\begin{align}\label{ineq:cardboundedges}
\#\mathrm{extreme}(W_{\varphi}) \leq 3^N\,,
\end{align}
where we recall $N$ defined in \eqref{def:edges}.
 Let us note that, due to the Krein-Milman Theorem (cf.~\cite{Brezis}, Theorem 1.13) and \eqref{eq:WphiL}, it is easy to see that there holds
\begin{align*}
\#\mathrm{extreme}(W_{\varphi}) = \#\mathrm{extreme}\left(L \left(V\cap\prod_{(i,j) \in \mathcal{E}}[0,c_{i,j}]\right) \right)\leq \#\mathrm{extreme}\left(V\cap \prod_{(i,j) \in \mathcal{E}}\left[0,c_{i,j}\right]\right)\,.
\end{align*}
In order to show \eqref{ineq:cardboundedges}, it remains to show
\begin{align}\label{ineq:cardboundintersection}
\#\mathrm{extreme}\left(V\cap \prod_{(i,j) \in \mathcal{E}}\left[0,c_{i,j}\right]\right)\leq 3^N \,.
\end{align}
\BBB In order to obtain this estimate we note that the extreme points of $V\cap \prod_{(i,j) \in \mathcal{E}}\left[0,c_{i,j}\right]$ lie on the $k$-dimensional (here $k$ is the co-dimension of the linear subspace $V$) facets of $\prod_{(i,j) \in \mathcal{E}}\left[0,c_{i,j}\right]$. Furthermore, we can find an injective relation between extreme points and these facets. In fact, if the matrix determining the intersection of the facet with $V$ is full rank, then the point of intersection is unique. If that is not the case, then the solution set is itself a subspace and one can add one additional condition to obtain the extreme point. This implies that here the extreme point is shared by more $k$-dimensional facets on a lower dimensional facet.  Note that there are at most $\binom{N}{k}2^{N-k}$ such facets and by the binomial formula we have that
\begin{align*}
\binom{N}{k}2^{N-k} \leq \sum_{j=0}^N \binom{N}{j}2^{N-j} = 3^N\,. 
\end{align*}\EEE
This concludes Step 3.
\end{step}
\end{proof}

\BBB
\section{Differentiability of the effective surface tension}\label{sec:dif}

In this Section, we prove Proposition~\ref{prop:dif} which states that $\varphi$ is
differentiable in totally irrational directions.
It is a corollary of the two lemmas which we state and prove below.
\begin{lemma}\label{lem:dif}
  Let  $\nu\in\mathbb{S}^{d-1}$,  let $u$ be a minimizer in~\eqref{eq:propvarphi} and assume that for any $s\in\R$,
  the set $\{u=s\}$ is finite. Then $\varphi$ is differentiable in $\nu$.
\end{lemma}
\begin{proof}
  The expression~\eqref{eq:dual representation} shows that $\varphi$ is a convex,
  one-homogeneous function with subgradient at $\nu$ given by
  \[
    \partial\varphi(\nu) = \left\{
      \frac{1}{T^d} \sum_{i \in \mathcal{L} \cap Q_T }\sum_{j \in \mathcal{L}}\alpha_{i,j}(i-j)\,:\,
      \alpha=(\alpha_{i,j})_{i,j}\in\mathcal{C} \textup{ maximizer in \eqref{eq:dual representation}}
      \right\}
  \]
  It is differentiable at $\nu$ if and only if the above set has exactly one element.
  
  Let $\alpha,\alpha'\in \mathcal{C}$ be two maximizers in~\eqref{eq:dual representation}. Classical
  optimality conditions guarantee that for any $i,j$, if $u(i)\neq u(j)$, then:
  \begin{equation}\label{eq:dualequal}
    \alpha_{i,j} =\alpha'_{i,j} = \begin{cases} c_{i,j} & \textup{ if }  u(i)-u(j)>0\\
      0 & \textup{ if } u(i)-u(j) < 0.
    \end{cases}
  \end{equation}
  Let us denote by $p,p'\in\partial\varphi(\nu)$ the subgradients given by the dual variables,
  respectively, $\alpha$ and $\alpha'$, we claim that $p=p'$. One has:
  \begin{equation}\label{eq:ppp}
    p-p' = \frac{1}{T^d}\sum_{i \in \mathcal{L} \cap Q_T }\sum_{j: u(j)=u(i)}
    (\alpha_{i,j}-\alpha'_{i,j})(i-j).
  \end{equation}
  
  Let $s\in\R$, $i_0\in\mathcal{L}\cap Q_T$ with $u(i_0)=s$ and such that the finite
  set $J_s:=\{j: u(j)=s\}$ has more than one element.
  For any $i,j$, let $\beta_{i,j}:=\alpha_{i,j}-\alpha'_{i,j}$.
  Then
  \begin{align*}
    \sum_{i\in J_s}\sum_{j\in J_s} \beta_{i,j} (i-j)
    &
      = \sum_{z\in\mathbb{Z}^d}\sum_{i\in J_s\cap (Tz+Q_T)} \sum_{j\in J_s}\beta_{i,j}(i-j) \\
    &= \sum_{z\in\mathbb{Z}^d} \sum_{i\in (J_s-Tz) \cap Q_T} \sum_{j\in J_s-Tz}\beta_{i,j}(i-j)
  \end{align*}
  where for the last line we have substituted $(i,j)$ with $(i-Tz,j-Tz)$ and
  used that $\beta$ is $Q_T$-periodic. In addition, we have that
  $u(i)=u(j)$ if and only if $u(i-Tz)=u(j-Tz)$ so that this can be rewritten:
  \[
    \sum_{i\in J_s}\sum_{j\in J_s} \beta_{i,j} (i-j)
    = \sum_{z\in\mathbb{Z}^d} \sum_{i\in (J_s-Tz) \cap Q_T} \sum_{j: u(j)=u(i)}\beta_{i,j}(i-j)
  \]
  By assumption, the sets $(J_s-Tz)\cap Q_T$, $z\in\mathbb{Z}^d$ are all disjoint.
  Otherwise, there would be $i,z$ with $s= u(i-Tz) = u(i) + T\langle \nu,z\rangle = s$,
  yielding in particular that $\langle\nu,z\rangle=0$, and one would deduce that
  $i-kTz\in J_s$ for all $k\in\mathbb{Z}$, a contradiction since we assumed $J_s$ was
  finite. As a consequence, showing that~\eqref{eq:ppp} vanishes is equivalent to
  showing that
  \begin{equation}\label{eq:level0}
    \sum_{i\in J_s}\sum_{j\in J_s} \beta_{i,j} (i-j)     = 0
  \end{equation}
  for any $s\in\mathbb{R}$ (such that $J_s$ is not empty and contains more than one point).
  Obviously, the expression in~\eqref{eq:level0} is also
  \[
     \sum_{i\in J_s}\sum_{j\in J_s} (\beta_{i,j}-\beta_{j,i}) i
  \]
  Thanks to the definition~\eqref{def:C} of $\mathcal{C}$, one has for any $i$ that
  $\sum_j \beta_{i,j}-\beta_{j,i}=0$, so that:
  \[
    \sum_{i\in J_s}\sum_{j\in J_s} (\beta_{i,j}-\beta_{j,i}) i =
    \sum_{i\in J_s}\sum_{j\not \in J_s} (\beta_{j,i}-\beta_{i,j}) i = 0
  \]
  thanks to~\eqref{eq:dualequal}. Hence, \eqref{eq:level0} holds and we deduce $p=p'$,
  which shows the lemma.  
\end{proof}
\begin{lemma}
  Let $\nu\in\mathbb{S}^{d-1}$ be totally irrational and let $u$ be a minimizer in~\eqref{eq:propvarphi}.
  Then for any $s\in\R$, the set $\{u=s\}$ is finite.
\end{lemma}
\begin{proof}
  Recalling the notation in the previous proof, let $s\in\mathbb{R}$ and consider the
  set $J_s:=\{u=s\}$. For $z\in\mathbb{Z}^d$, let $J_s^z = J_s\cap (Q_T+Tz)-Tz\subset Q_T$.
  For $i\in J_s^z$, $u(i) = s+ T\langle z,\nu\rangle$. Since $\nu$ is totally irrational,
  we deduce that $J_s^z\cap J_s^{z'} =\emptyset$ for any $z\neq z'$, showing that 
 all sets $J_s^z$ but a finite number must be empty. Hence $J_s$ is finite.
\end{proof}\EEE

\section{Numerical illustration}

\subsection{A simplified framework}
In this section, we address, as an illustrative experiment, the following
issue. We consider a basic $2D$ cartesian graph $\{(i,j): 0\le i\le M-1,
0\le j\le N-1\}$, representing for instance the pixels of an image,
and we want to approximate on this discrete grid the two-dimensional
total variation $\int_{\Omega} |Du|$, $u\in BV(\Omega)$. Here it is assumed
that $\Omega\subset \R^2$ is a rectangle and that
$\{0,\dots, M-1\}\times \{0,\dots, N-1\}$ is a discretization of
$\Omega$ at a length scale $\sim 1/N\sim 1/M$.

There are of course many ways to do this, but we propose here
to consider a family of discrete ``graph'' total variations, defined
for a \BBB family \EEE  $(u_{i,j})_{i,j}\in\R^{M\times N}$ by:
\begin{equation}\label{eq:defdiscreteTV}
  \begin{aligned}
    J(u)=  \sum_{i,j} c^+_{i+\ha,j}&(u_{i+1,j}-u_{i,j})^++
    c^-_{i+\ha,j}(u_{i,j}-u_{i+1,j})^+
    \\
    & + c^+_{i,j+\ha}(u_{i,j+1}-u_{i,j})^++
    c^-_{i,j+\ha}(u_{i,j}-u_{i,j+1})^+
  \end{aligned}
\end{equation}
and  which involves only nearest-neighbour interactions
in horizontal and vertical directions.

We assume in addition that the weight $c^\pm$ are $T$-periodic
for some $T\in\mathbb{N}$, $T>0$, that is,
$  C^\pm_{a+kT,b+lT} = c^\pm_{a,b}$ for any $(k,l)\in\mathbb{Z}^2$,
$(a,b)=(i+\ha,j)$ or $(i,j+\ha)$, as long as the points
fall inside the grid.

For $T=1$, $c^\pm_{a,b}\equiv 1$, it is standard that~\eqref{eq:defdiscreteTV}
approximates, in the continuum limit, the anisotropic total
variation $\int_\Omega |\partial_1 u|+|\partial_2 u|$, which,
if used for instance as a regularizer for image denoising
or reconstruction,
may produce undesired artefacts (although hardly visible
on standard applications, see Figure~\ref{fig:Denoise}).

A standard way to mitigate this issue (besides, of course,
resorting to numerical analysis based on finite
differences or elements in order to define
more refined discretizations), 
is to add to~\eqref{eq:defdiscreteTV} diagonal interactions, with
appropriate weights, in order to improve the isotropy of
the limit (see for instance~\cite{BoykovKolmogorov2003}), with the drawback
of complexifying the graph and the optimization.
We show here that a similar effect can be attained by homogenization.
To illustrate this, let us first consider the simplest
situation, for $T=2$.

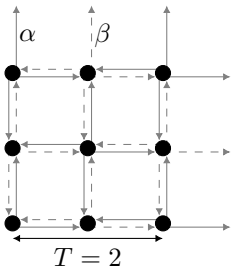
\begin{figure}[htb]
  \begin{tikzpicture}

    \tikzset{>={Latex[width=1mm,length=1mm]}};
    
    \foreach \j in {0,2}{
      \foreach \i in {0,2}{
        \draw[->,gray](\i,\j-.05)--++(0:.9);
      }
    }
    \foreach \j in {1}{
      \foreach \i in {0,2}{
        \draw[dashed,->,gray](\i,\j-.05)--++(0:.9);
      }
    }
    \foreach \j in {0,2}{
      \foreach \i in {1}{
        \draw[dashed,->,gray](\i,\j-.05)--++(0:.9);
      }
    }
    \foreach \j in {1}{
      \foreach \i in {1}{
        \draw[->,gray](\i,\j-.05)--++(0:.9);
      }
    }

    \foreach \j in {0,2}{
      \foreach \i in {0}{
        \draw[dashed,->,gray](\i+1,\j+.05)--++(180:.9);
      }
    }
    \foreach \j in {0,2}{
      \foreach \i in {1}{
        \draw[->,gray](\i+1,\j+.05)--++(180:.9);
      }
    }
    \foreach \j in {1}{
      \foreach \i in {1}{
        \draw[dashed,->,gray](\i+1,\j+.05)--++(180:.9);
      }
    }
    \foreach \j in {1}{
      \foreach \i in {0}{
        \draw[->,gray](\i+1,\j+.05)--++(180:.9);
      }
    }

    \foreach \j in {0,2}{
      \foreach \i in {0,2}{
        \draw[->,gray](\i+.05,\j)--++(90:.9);
      }
    }
    \foreach \j in {1}{
      \foreach \i in {0,2}{
        \draw[dashed,->,gray](\i+.05,\j)--++(90:.9);
      }
    }
    \foreach \j in {0,2}{
      \foreach \i in {1}{
        \draw[dashed,->,gray](\i+.05,\j)--++(90:.9);
      }
    }
    \foreach \j in {1}{
      \foreach \i in {1}{
        \draw[->,gray](\i+.05,\j)--++(90:.9);
      }
    }

    \foreach \j in {0}{
      \foreach \i in {0,2}{
        \draw[dashed,->,gray](\i-.05,\j+1)--++(270:.9);
      }
    }
    \foreach \j in {0}{
      \foreach \i in {1}{
        \draw[->,gray](\i-.05,\j+1)--++(270:.9);
      }
    }
    \foreach \j in {1}{
      \foreach \i in {0,2}{
        \draw[->,gray](\i-.05,\j+1)--++(270:.9);
      }
    }
    \foreach \j in {1}{
      \foreach \i in {1}{
        \draw[dashed,->,gray](\i-.05,\j+1)--++(270:.9);
      }
    }

    \foreach \j in {0,1,2}{
      \foreach \i  in {0,1,2}{
        \draw[fill=black] (\i,\j) circle(.1);
      }
    }

    \draw(0.2,2.5) node{$\alpha$};
    \draw(1.2,2.5) node{$\beta$};

    \draw[<->](0,-.2)--++(0:2);
    \draw(1,-.2) node[anchor=north]{$T=2$};

\end{tikzpicture}

\caption{The alternating $2$-periodic coefficients yielding
  the smallest anisotropy}\label{fig:coefficients}
\end{figure}
In that case, one can explicitly build coefficients $c^\pm_{a,b}$,
taking two values $\alpha,\beta$ (see Figure~\ref{fig:coefficients}),
which will yield the homogenized surface tension
\begin{equation}\label{eq:ST2}
  \varphi(\nu) = (\sqrt{2}-1)\left(|\nu_1|+|\nu_2|
  +\frac{|\nu_1+\nu_2|}{\sqrt{2}}  +\frac{|\nu_1-\nu_2|}{\sqrt{2}}\right)
\end{equation}
whose $1$-level set (or \textit{Frank diagram}) is shown
in Figure~\ref{fig:FD2}.
\begin{figure}[htb]
  \begin{center}
    \begin{tikzpicture}[scale=2,declare function={
        phi(\x,\y) = 0.4142136*(abs(\x)+abs(\y)
        +(abs(\x+\y)+abs(\x-\y))/1.4142136); }]
      \draw[lightgray,very thin] 
      [step=0.25cm] (-1,-1) grid (1,1);
      \draw[->] (-1.25,0) -- (1.25,0) node[right] {$\nu_1$};
      \draw[->] (0,-1.25) -- (0,1.25) node[right] {$\nu_2$};

      \foreach \pos in {-1,1}
      \draw[shift={(\pos,0)}] (0pt,2pt) -- (0pt,-2pt) node[below] {$\pos$};
      \foreach \pos in {-1,1}
      \draw[shift={(0,\pos)}] (2pt,0pt) -- (-2pt,0pt) node[left] {$\pos$};
      

      \draw[thick,variable=\t,domain=0:360,samples=720]
      plot ({cos(\t)/phi(cos(\t),sin(\t))},{sin(\t)/phi(cos(\t),sin(\t))});
      
    \end{tikzpicture}
  \end{center}
  \caption{The Frank diagram $\{\nu:\varphi(\nu)\le 1\}$
    given by~\eqref{eq:ST2}} \label{fig:FD2}
\end{figure}
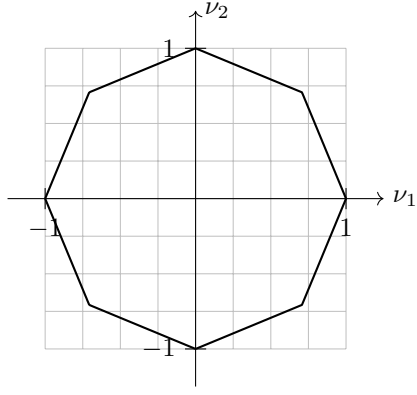
Observe that this is the same anisotropy which
would be obtained by using constant coefficients and adding
interactions along the edges $((i,j),(i+1,j+1))$ and $((i,j),(i+1,j-1))$.

In order to obtain~\eqref{eq:ST2}, one needs to tune $\alpha,\beta$ so
that a vertical edge and a diagonal edge, in the most favorable position,
have the same length (with a $\sqrt{2}$ factor for the diagonal, whose
intersection with the periodicity cell is of course longer).
This is ensured if $\alpha+\beta = 4\alpha/\sqrt{2}$, that is,
$\beta = (2\sqrt{2}-1)\alpha$. We find that choosing
\begin{equation}\label{eq:optab2}
  \begin{cases} \alpha =\frac{1}{4\sqrt{2}} \approx 0.1768\\
    \beta = (2\sqrt{2}-1)\alpha \approx 0.3232
  \end{cases}
\end{equation}
yields~\eqref{eq:ST2}, as an effective homogenized anisotropy.

For larger periodicity cells, it seems difficult to do a similar analysis,
first of all, because one should not expect the optimal minimizers, in
most directions (if not all), to be given by straight lines, but rather
by periodic perturbations of straigth lines. We propose an optimization
process in order to compute the optimal weights $c^\pm_{a,b}$.

\subsection{The optimization method}

The effective surface tension is obtained by solving the cell
problem:
\begin{equation}\label{eq:cellp}
  \begin{aligned}
    \phi(\nu) = \min_u \Big\{\sum_{(i,j)\in Y}&c^+_{i+\ha,j}(u_{i+1,j}-u_{i,j})^++
    c^-_{i+\ha,j}(u_{i,j}-u_{i+1,j})^+
    \\
    +\ & c^+_{i,j+\ha}(u_{i,j+1}-u_{i,j})^++
    c^-_{i,j+\ha}(u_{i,j}-u_{i,j+1})^+\ :
    \\ & u_{i,j}- \nu \cdot {\begin{pmatrix} i \\ j \end{pmatrix}}\ \textup{$Y$-periodic}
    \ \Big\}
  \end{aligned}
\end{equation}
where $Y=\mathbb{Z}^2\cap ([0,T)\times [0,T))$ is the periodicity cell.
This is easily solved, for instance by a saddle-point algorithm~\cite{ChaPock-JMIV} which aims at finding a solution to:
\begin{multline*}
  \phi(\nu) =
  \min_{\BBB v \, Y\text{-periodic}\EEE} \max_{0 \le w^\pm_{\bullet}\le 1}
  \sum_{(i,j)\in Y}
  (w^+_{i+\ha,j}c^+_{i+\ha,j}- w^-_{i+\ha,j}c^-_{i+\ha,j})(v_{i+1,j}-v_{i,j}+\nu_1)
\\  +  (w^-_{i,j+\ha}c^-_{i,j+\ha}- w^+_{i,j+\ha}c^+_{i,j+\ha})(v_{i,j+1}-v_{i,j}+\nu_2),
\end{multline*}
where we have replaced the variable $u$ with
the periodic vector $v_{i,j}=u_{i,j}-\nu\cdot(i,j)^T$.
For technical reasons, we need to ``regularize'' slightly this problem
in order to make it differentiable with respect to the coefficients
$\mathbf{c}=(c^\pm_{\bullet})$. This is done by introducing $\e>0$ a (very) small
parameter and adding to the previous objective the penalization
\[
  - \frac{\e}{2} \sum_{(i,j)\in Y}
  (w^+_{i+\ha,j})^2 + (w^-_{i+\ha,j})^2 + (w^-_{i,j+\ha})^2 + (w^+_{i,j+\ha})^2
  + \frac{\e}{2} \sum_{(i,j)\in Y} v_{i,j}^2
\]
which makes the problem strongly convex/concave and the solutions $w,v$
unique. We call $\phi_\e(\nu)[\mathbf{c}]$ the corresponding value.
The advantage of this regularization is that one can easily show that
$\mathbf{c}\mapsto\phi_\e(\nu)[\mathbf{c}]$ is locally $C^{1,1}$, with a
gradient given by:
\begin{multline*}
  \lim_{t\to 0}\frac{\phi_\e(\nu)[\mathbf{c}+t\mathbf{d}]-
  \phi_\e(\nu)[\mathbf{c}]}{t} =
  \sum_{(i,j)\in Y}
  (w^+_{i+\ha,j}d^+_{i+\ha,j}- w^-_{i+\ha,j}d^-_{i+\ha,j})(v_{i+1,j}-v_{i,j}+\nu_1)
 \\ +  (w^-_{i,j+\ha}d^-_{i,j+\ha}- w^+_{i,j+\ha}d^+_{i,j+\ha})(v_{i,j+1}-v_{i,j}+\nu_2)
\end{multline*}
where $(w,v)$ solves the saddle-point problem which defines $\phi_\e(\nu)[\mathbf{c}]$.

Then, to find coefficients which ensure that $\phi$ is as ``isotropic''
as possible, one fixes a finite set of directions
$(\nu_1,\dots,\nu_k)$ (typically, $(\cos (2\ell\pi/k),\sin (2\ell\pi/k))$
for $\ell=1,\dots,k$), and \BBB uses \EEE a first
order gradient descent algorithm to optimize:
\[
  \mathscr{L}(\mathbf{c}) = \sum_{\ell=1}^k \left(\phi_\e(\nu_\ell)[\mathbf{c}]
  -1\right)^2
\]
The problem is easily solved for $Y=\{0,1\}\times\{0,1\}$,
$k=8$ and $(\nu_\ell)_{\ell=1}^8$ given as above. For larger periodicity
cells and more directions, it easily gets trapped in local minima and we
use a random initialization in order to be able to find satisfactory
solutions. We then test the result by computing
the un-regularized surface tension $\phi$ with the resulting coefficients
$\mathbf{c}$. We show some results in the next section.
Of course, taking a large value of $\e$ will make the problem easier
to solve, but the learned coefficients will not allow to reconstruct
a satisfactory surface tension: we need to choose $\e$ small, an order
of magnitude below the error which we expect on the anisotropy of $\phi$.

\subsection{Numerical results}

We show the outcome of the optimization, in the
periodicity cell $Y=\{0,\dots T-1\}\times \{0,\dots,T-1\}$ for $T=2,4,6,8$.
We plot first the set $\{\varphi\le 1\}$ or Frank diagram for the effective
surface tensions.
\begin{figure}[!htbp]
  \includegraphics[width=.3\textwidth]{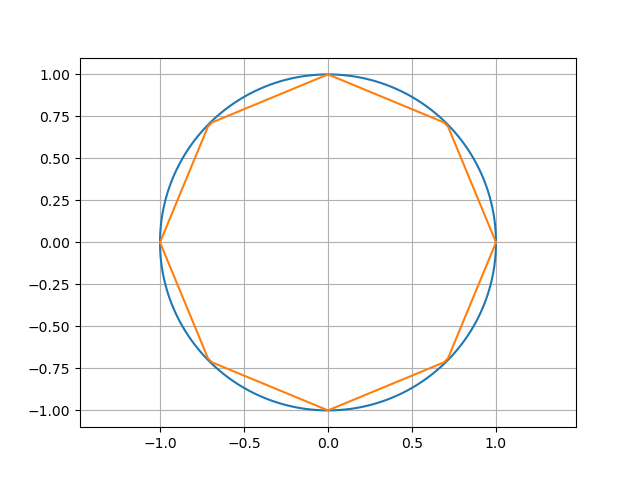}
  \includegraphics[width=.3\textwidth]{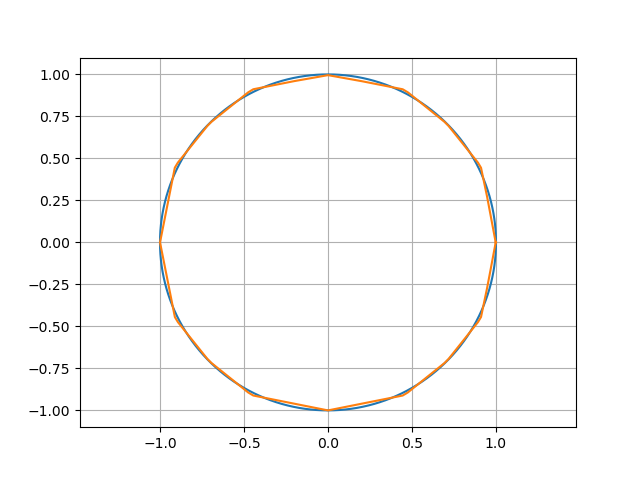}
  \includegraphics[width=.3\textwidth]{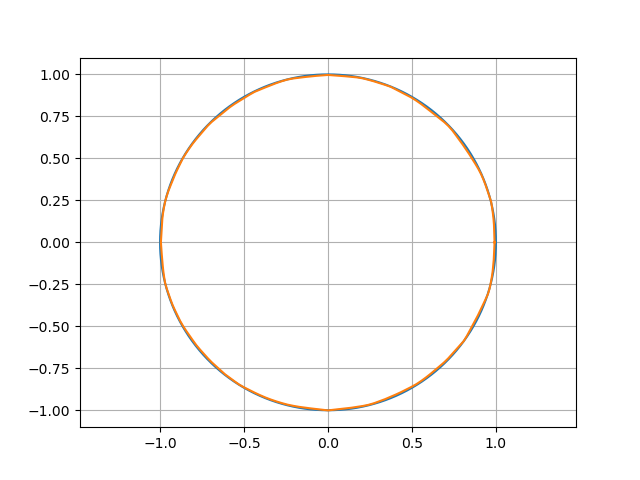}
  \caption{Frank diagrams of the effective anisotropies for $T=2,4,8$.}
  \label{fig:Frank}
\end{figure}
Figure~\ref{fig:Frank} shows the diagram obtained,
for $T=2,4,8$. For $T=2$, the optimization yields the same
anisotropy as our \BBB construction in \eqref{eq:optab2} which gives the anistropy \eqref{eq:ST2} and which we conjecture to be optimal for $\mathscr{L}(\mathbf{c})$ taking $k=8$ \EEE (compare with Fig.~\ref{fig:FD2}).
However, except when initialized with the values in~\eqref{eq:optab2},
the algorithm usually outputs different values with the
same effective anisotropy, see Fig.~\ref{fig:graph2}
(the values in~\eqref{eq:optab2} are in some sense better, as for instance
a vertical edge will always have the same effective energy with these values,
while with the \BBB computed \EEE values displayed in Fig.~\ref{fig:graph2}, it
will need to pass through the edges in the second column of the cell in
order to get the minimal energy).
%
%

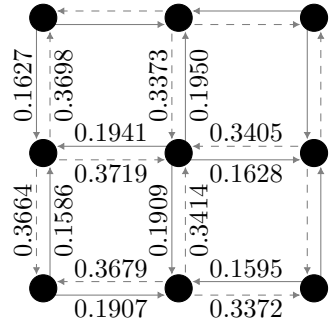
\begin{figure}[!htbp]
  \begin{tikzpicture}[scale=1.8]

    \tikzset{>={Latex[width=1mm,length=1mm]}};
    
    \foreach \j in {0,2}{
      \foreach \i in {0}{
        \draw[->,gray](\i,\j-.05)--++(0:.9);
      }
    }
    \foreach \j in {1}{
      \foreach \i in {0}{
        \draw[dashed,->,gray](\i,\j-.05)--++(0:.9);
      }
    }
    \foreach \j in {0,2}{
      \foreach \i in {1}{
        \draw[dashed,->,gray](\i,\j-.05)--++(0:.9);
      }
    }
    \foreach \j in {1}{
      \foreach \i in {1}{
        \draw[->,gray](\i,\j-.05)--++(0:.9);
      }
    }

    \foreach \j in {0,2}{
      \foreach \i in {0}{
        \draw[dashed,->,gray](\i+1,\j+.05)--++(180:.9);
      }
    }
    \foreach \j in {0,2}{
      \foreach \i in {1}{
        \draw[->,gray](\i+1,\j+.05)--++(180:.9);
      }
    }
    \foreach \j in {1}{
      \foreach \i in {1}{
        \draw[dashed,->,gray](\i+1,\j+.05)--++(180:.9);
      }
    }
    \foreach \j in {1}{
      \foreach \i in {0}{
        \draw[->,gray](\i+1,\j+.05)--++(180:.9);
      }
    }

    \foreach \j in {0}{
      \foreach \i in {0,2}{
        \draw[->,gray](\i+.05,\j)--++(90:.9);
      }
    }
    \foreach \j in {1}{
      \foreach \i in {0,2}{
        \draw[dashed,->,gray](\i+.05,\j)--++(90:.9);
      }
    }
    \foreach \j in {0}{
      \foreach \i in {1}{
        \draw[dashed,->,gray](\i+.05,\j)--++(90:.9);
      }
    }
    \foreach \j in {1}{
      \foreach \i in {1}{
        \draw[->,gray](\i+.05,\j)--++(90:.9);
      }
    }

    \foreach \j in {0}{
      \foreach \i in {0,2}{
        \draw[dashed,->,gray](\i-.05,\j+1)--++(270:.9);
      }
    }
    \foreach \j in {0}{
      \foreach \i in {1}{
        \draw[->,gray](\i-.05,\j+1)--++(270:.9);
      }
    }
    \foreach \j in {1}{
      \foreach \i in {0,2}{
        \draw[->,gray](\i-.05,\j+1)--++(270:.9);
      }
    }
    \foreach \j in {1}{
      \foreach \i in {1}{
        \draw[dashed,->,gray](\i-.05,\j+1)--++(270:.9);
      }
    }

    \foreach \j in {0,1,2}{
      \foreach \i  in {0,1,2}{
        \draw[fill=black] (\i,\j) circle(.1);
      }
    }

%
%

    \draw(.5,.15) node{$0.3679$};
    \draw(.5,-.15) node{$0.1907$};
    \draw(1.5,.15) node{$0.1595$};
    \draw(1.5,-.15) node{$0.3372$};
    \draw(.5,1.15) node{$0.1941$};
    \draw(.5,.85) node{$0.3719$};
    \draw(1.5,1.15) node{$0.3405$};
    \draw(1.5,.85) node{$0.1628$};
    \draw(-0.15,.5) node[rotate=90]{$0.3664$};
    \draw(0.15,.5) node[rotate=90]{$0.1586$};
    \draw(-0.15,1.5) node[rotate=90]{$0.1627$};
    \draw(0.15,1.5) node[rotate=90]{$0.3698$};
    \draw(.85,.5) node[rotate=90]{$0.1909$};
    \draw(1.15,.5) node[rotate=90]{$0.3414$};
    \draw(.85,1.5) node[rotate=90]{$0.3373$};
    \draw(1.15,1.5) node[rotate=90]{$0.1950$};


\end{tikzpicture}

\caption{An example of optimized $2$-periodic coefficients yielding
  the same anisotropy as the choice~\eqref{eq:optab2}}\label{fig:graph2}
\end{figure}
\begin{figure}[!htbp]
  \includegraphics[width=.25\textwidth]{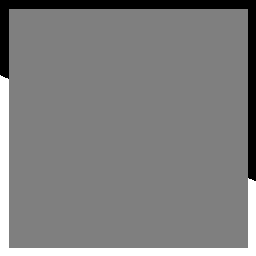} 
  \hspace{2mm}
  \includegraphics[width=.25\textwidth]{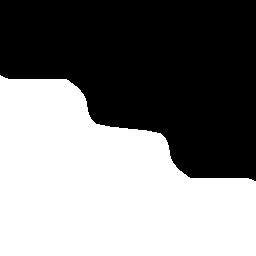}\hspace{2mm}
  \includegraphics[width=.25\textwidth]{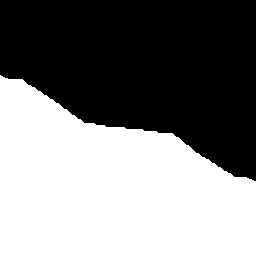}\\
  \includegraphics[width=.25\textwidth]{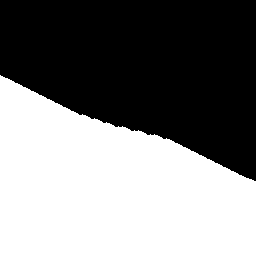}\hspace{2mm}
  \includegraphics[width=.25\textwidth]{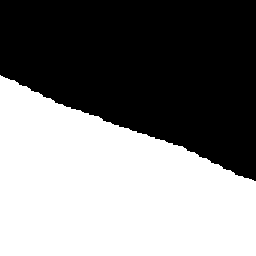}\hspace{2mm}
  \includegraphics[width=.25\textwidth]{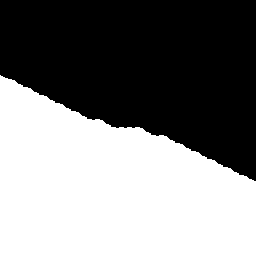}
  \caption{A minimal half-plane in the orientation
    $(\cos 3\pi/8,\sin 3\pi/8)$. Top left, boundary datum, the region
    where the perimeter is minimized is in gray. Top, middle:
    $\varphi(\nu)=|\nu_1|+|\nu_2|$. Top, right: optimal effective $\varphi$ for $T=2$.
    Bottom: for $T=4$, $6$, $8$.}  \label{fig:Inpaint}
\end{figure}
For $T=4$, one sees that the behavior is almost isotropic, while for
$T=8$, the relative error with the perfect unit disk is  about $1\%$. \BBB
Here, we estimated this error as
$(\max_\ell\phi(\nu_\ell)-\min_\ell\phi(\nu_\ell))/\min_{\ell} \phi(\nu_\ell)$,
where  $\ell \in \{1,\ldots,k\}$ and $k=180$. \EEE
We illustrate this on an ``inpainting'' example, which consists in finding the
minimal line in a given direction. 
We consider as an example the direction $(\cos 3\pi/8,\sin 3\pi/8)$, which is
irrational, so that there cannot be a fully periodic solution. The figure~\ref{fig:Inpaint}
displays several minimal half-planes in this orientation. Observe that for this orientation,
the results for $T=4$ or $6$ look nicer than the result obtained for $T=8$.

We also show a denoising example based on the ``ROF'' method (which
consists simply in minimizing the total variation (defined
by the surface tension $\varphi$) of an image with a quadratic
penalization of the distance to a noisy data, in order to produce a
denoised version, see~\cite{RudinOsherFatemi}) 
with the anisotropic tension $\varphi(\nu)=|\nu_1|+|\nu_2|$ (``$T=1$'')
and the optimized homogenized surface tension for $T=4$.
The original image is degraded with a Gaussian noise with $10\%$  standard
deviation (with respect to the range of the values).
Here, the difference between the two regularizers is hardly perceptible (since the data term strongly influences the position of
the discontinuities), yet a close-up (bottom row) allows to see a slight difference, for instance on the cheek where the $T=1$ anisotropy
produces block structures.
\begin{figure}[!htbp]
  \includegraphics[width=.25\textwidth]{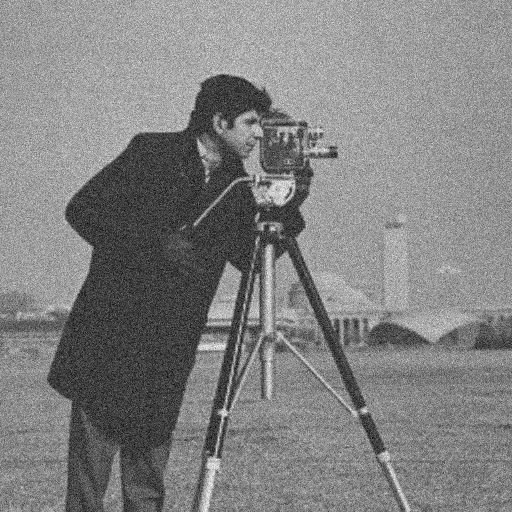}\hspace{2mm}
  \includegraphics[width=.25\textwidth]{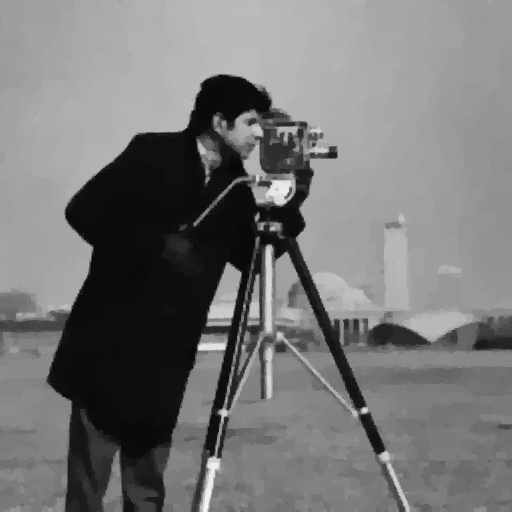}\hspace{2mm}
  \includegraphics[width=.25\textwidth]{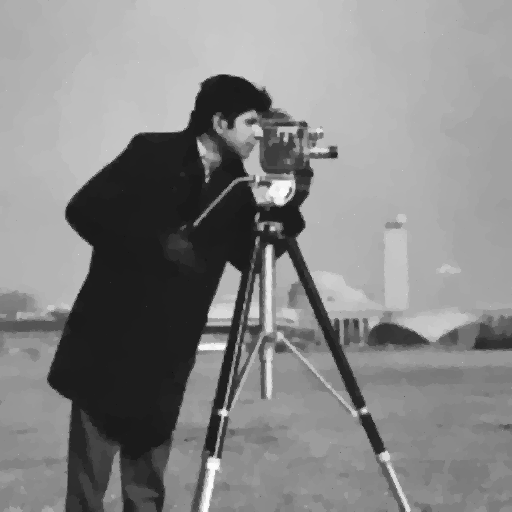}\\[2mm]
  \includegraphics[width=.25\textwidth,trim=6cm 10.5cm 7cm 2.5cm,clip]{camera_noise}\hspace{2mm}
  \includegraphics[width=.25\textwidth,trim=6cm 10.5cm 7cm 2.5cm,clip]{camera_denoise_l1}\hspace{2mm}
  \includegraphics[width=.25\textwidth,trim=6cm 10.5cm 7cm 2.5cm,clip]{camera_denoise_4}
  \caption{``ROF'' denoising example. Left: noisty image.
    Middle, denoised with
    $\varphi(\nu)=|\nu_1|+|\nu_2|$. Right: with the effective
    tension computed for $T=4$.}  \label{fig:Denoise}
\end{figure}

\section*{Acknowledgements} 
This work was supported by the Deutsche Forschungsgemeinschaft (DFG, German Research Foundation) under Germany's Excellence Strategy EXC 2044 -390685587, Mathematics M\"unster: Dynamics--Geometry--Structure. The authors thank
the reviewers for their careful reading of the paper and comments
which led to much improvement with respect to the first version.

\end{document}